\documentclass[conference]{IEEEtran}
\IEEEoverridecommandlockouts
\usepackage{cite}
\usepackage{amsmath,amssymb,amsfonts}
\usepackage{graphicx}
\usepackage{textcomp}
\usepackage{xcolor}
\usepackage{hyperref}
\usepackage{algorithm}
\usepackage{algpseudocode}
\usepackage{amsthm}

\newtheorem{thm}{Theorem}[section]
\newtheorem{lem}[thm]{Lemma}
\newtheorem{prop}[thm]{Proposition}
\newtheorem{cor}{Corollary}

\newtheorem{assum}{Assumption}[section]
\newtheorem{cond}{Condition}[section]
\newtheorem{rem}{Remark}

\def\P{{\mathbb P}}
\def\E{{\mathbb E}}

\def\R{\mathbb R}
\def\Z{\mathbb Z}
\usepackage{bm}
\renewcommand\vec[1]{{\bm #1}}

\def\BibTeX{{\rm B\kern-.05em{\sc i\kern-.025em b}\kern-.08em
    T\kern-.1667em\lower.7ex\hbox{E}\kern-.125emX}}
    
\newcommand{\linebreakand}{%
      \end{@IEEEauthorhalign}
      \hfill\mbox{}\par
      \mbox{}\hfill\begin{@IEEEauthorhalign}
}
 
\begin{document}

\title{Detection of Sparse Anomalies in \\ High-Dimensional Network Telescope Signals\\
\thanks{This work was partially supported by the National Science Foundation under awards CNS-1823192 and CNS-2120400.}
}

\author{\IEEEauthorblockN{Rafail Kartsioukas}
\IEEEauthorblockA{
% \textit{Department of Statistics} \\
% \textit{University of Michigan, Ann Arbor} \\
% Ann Arbor, USA \\
rkarts@umich.edu}
\and
\IEEEauthorblockN{Rajat Tandon}
\IEEEauthorblockA{
% \textit{Information Sciences Institute} \\
% \textit{University of Southern California}\\
% Marina del Rey, USA \\
rajattan@usc.edu}
\and
\IEEEauthorblockN{Zheng Gao}
\IEEEauthorblockA{
% \textit{Department of Statistics} \\
% \textit{University of Michigan, Ann Arbor} \\
% Ann Arbor, USA \\
gaozheng@umich.edu }\\
%\linebreakand
\and
\IEEEauthorblockN{Jelena Mirkovic}
\IEEEauthorblockA{
% \textit{Information Sciences Institute} \\
% \textit{University of Southern California }\\
% Marina del Rey, USA \\
mirkovic@isi.edu }
\and
\IEEEauthorblockN{Michalis Kallitsis}
\IEEEauthorblockA{
% \textit{Research \& Development} \\
% \textit{Merit Network, Inc.}\\
% Ann Arbor, USA \\
mgkallit@merit.edu}
\and
\IEEEauthorblockN{Stilian Stoev}
\IEEEauthorblockA{
% \textit{Department of Statistics} \\
% \textit{University of Michigan, Ann Arbor}\\
% Ann Arbor, USA \\
sstoev@umich.edu}
}

\maketitle

\begin{abstract}
Network operators are
increasingly overwhelmed with incessant cyber-security threats,
ranging from malicious network reconnaissance to attacks such 
as distributed denial of service and data breaches.
A large number of these attacks could be prevented if the network
operators were better equipped
with threat intelligence information that would allow them 
to block nefarious scanning activities. Network telescopes or ``darknets" offer
a unique window into observing Internet-wide scanners and other malicious
entities, and they could offer early warning signals to operators
that would be critical for infrastructure protection.
% A network telescope consists of unused or ``dark" IP spaces that 
% serve no users, and solely passively observes
% any Internet traffic destined to the ``telescope sensor" in an attempt to record
% ubiquitous network scanners, malware that forage for vulnerable
% devices, and other dubious activities. 
Hence, monitoring network telescopes for timely detection of 
coordinated and heavy scanning activities is an important, albeit challenging, task. 
The challenges mainly arise due to the non-stationarity and the dynamic nature
of Internet traffic and, more importantly, the fact that one needs to monitor high-dimensional signals (e.g., all TCP/UDP
ports) to search for ``sparse" anomalies. We propose statistical methods to address both challenges
in an efficient and ``online" manner; our work is validated both with synthetic data and with
real-world data from a large network telescope.
\end{abstract}


\section{Introduction}

The Internet has evolved into a complex ecosystem, comprised
of a plethora of network-accessible services and end-user devices
that are frequently mismanaged, not properly maintained and secured,
and outdated with untreated software vulnerabilities.  Adversaries
are increasingly becoming aware of this ill-secure Internet 
landscape and leverage it to their advantage for launching attacks against critical infrastructure. 
Examples abound: foraging for mismanaged NTP and DNS open resolvers (or other UDP-based services)
is a well-known attack vector that can be exploited to incur volumetric reflection-and-amplification
distributed denial of service (DDoS) attacks~\cite{10.1145/2663716.2663717, AmplificationHell, Kuhrer:2014:EHR:2671225.2671233};
searching for and compromising insecure Internet-of-Things (IoT) devices (such as home
routers, Web cameras, etc.) has led to the outset of the Mirai botnet back in 2016
that was responsible for some of the largest DDoS ever recorded~\cite{203628, Krebbs, OVH, dyn};
variants of the Mirai epidemic still widely circulate (e.g., the Mozi~\cite{mozi2020} and Meri~\cite{krebbs2021} botnets)
and assault services at a global scale; cybercriminals have been exploiting the COVID-19 pandemic
to infiltrate networks via insecure VPN teleworking technologies that have been
deployed to facilitate work-from-home opportunities~\cite{cisa-covid19}. 

Notably, the initial phase of the aforementioned attacks is \emph{network scanning}, a step that is
necessary to detect and afterwards exploit vulnerable services/hosts. Against this background, 
network operators are tasked with monitoring and protecting
their networks and germane services utilized by their users. While many enterprises and large-scale
networks operate sophisticated firewalls and intrusion detection systems (e.g., Zeek, Suricata or
other non-open-source solutions), early signs of malicious network scanning activities
may not be easily noticed from their vantage points. 
Large \emph{Network Telescopes} or \emph{Darknets}~\cite{caida_telescope_report, orion}, however,
can fill this gap and can 
provide \emph{early warning notifications and insights} for emergent network threats to security analysts. 
Network telescopes consist of monitoring infrastructure that receives and records
unsolicited traffic destined to vast swaths of \emph{unused} but \emph{routed} Internet address spaces (i.e., millions
of IPs). This traffic, coined as ``Internet Background Radiation"~\cite{10.1145/1028788.1028794, 10.1145/1879141.1879149},
captures traffic from nefarious actors that perform Internet-wide scanning activities,
malware and botnets that aim to infect other victims, ``backscatter" activities that
denote DoS attacks~\cite{10.1145/1879141.1879149}, etc. Thus, Darknets
offer a unique lens into macroscopic Internet activities and timely detection
of new abnormal Darknet behaviors is extremely important.

In this paper, we consider Darknet data from the ORION Network Telescope
operated by Merit Network, Inc.~\cite{orion}, and construct multivariate signals
for various TCP/UDP ports (as well as other types of traffic) that denote the amount
of packets sent to the Darknet towards a particular port per monitoring interval (e.g., minutes). (See Figure~\ref{fig:illustration}.)
Our goals are to \emph{detect} when an ``anomaly" occurs 
in the Darknet\footnote{All traffic captured in the Darknet can be considered ``anomalous" since Darknets
serve no real services; however, henceforth
we slightly abuse the terminology and refer to traffic ``anomalies" in the statistical sense.}~and to also
accurately \emph{identify} the culprit port(s);
such threat intelligence would be invaluable in 
diagnosing emerging new vulnerabilities (e.g., ``zero-day" attacks).
%and would assist network operators in better defending their networks.
Our algorithms are based on the state-of-the-art theoretical results on the~\emph{sparse signal support recovery} 
problem in a high-dimensional setting (see, e.g., \cite{10.3150/20-BEJ1197} and the recent 
monograph  \cite{gao:stoev:2021}).
\textbf{Our main contributions are}: 1) we showcase, using simulated as well as real-world Darknet data,
that signal trends (e.g., diurnal or weekly scanning patterns) can be filtered out from
the multi-variate scanning signals using efficient \emph{sequential PCA (Principal Component Analysis)} techniques~\cite{6483308};
2) using recent theory~\cite{10.3150/20-BEJ1197}, we demonstrate that simple thresholding techniques applied individually 
on each univariate time-series exhibit better detection power than competing methods proposed in the literature
for diagnosing network anomalies~\cite{10.1145/1030194.1015492}; 
3) we propose and apply a non-parametric approach as a thresholding mechanism for the
recovery of sparse anomalies; and
4) we illustrate our methods on real-world Darknet data using
techniques amenable to online/streaming implementation.

\section{Related Work}

Network telescopes have been widely employed by the networking
community to understand various macroscopic Internet events. E.g.,
Darknet data helped to shed light into botnets~\cite{203628, 6717049},
to obtain insights about network outages~\cite{Benson:2012:GIA:2413247.2413285, Dainotti:2012:EBH:2096149.2096154},
to understand denial of service attacks~\cite{Moore:2006:IID:1132026.1132027, Jonker:2017:MTU:3131365.3131383, 10.1145/2663716.2663717},
for examining the behavior of IoT devices~\cite{8450404}, for observing Internet misconfigurations~\cite{Czyz:2013:UII:2504730.2504732,10.1145/1879141.1879149}, etc.

Mining of meaningful patterns in Darknet data is a challenging task due
to the dimensionality of the data and the heterogeneity of the ``Darknet features" that one
could invoke. Several studies have resorted to unsupervised machine learning techniques, such
as clustering, for the task at hand. 
Niranjana et al.~\cite{niranjana2020darknet} propose using Mean Shift clustering algorithms on TCP features
to cluster source IP addresses to find attack patterns in Darknet traffic. % clustering
Ban et al.~\cite{ban2012behavior} present a monitoring system that characterizes
the behavior of long term cyber-attacks by mining Darknet traffic. In this system, machine learning techniques such as clustering, classification and function regression are applied to the analysis of Darknet traffic.  % clustering
Bou-Harb et al.~\cite{bou2014multidimensional} propose
a multidimensional monitoring method for source port 0 probing attacks by analyzing Darknet
traffic. By performing unsupervised machine learning techniques on Darknet traffic, 
the activities by similar types of hosts are grouped by employing a set of statistical-based
behavioral analytics. This approach is targeted only for source port 0 probing attacks. % clustering, unsupervised learning
Nishikaze et al.~\cite{nishikaze2015large} present a machine learning approach for large-scale monitoring of malicious activities on Internet using Darknet traffic. In the proposed system, network packets sent from a subnet to a Darknet are collected, and 
they are transformed into 27-dimensional TAP (Traffic
Analysis Profile) feature vectors. Then, a hierarchical clustering is performed to obtain clusters 
for typical malicious behaviors. In the monitoring phase, the malicious activities in a subnet are 
estimated from the closest TAP feature cluster. Then, such TAP feature clusters for all subnets are 
visualized on the proposed monitoring system in real time to identify malicious activities. % clustering
Ban et al.~\cite{ban2016towards} present a study on early detection of emerging novel attacks. The 
authors identify attack patterns on Darknet data using a clustering algorithm and perform nonlinear 
dimension reduction to provide visual hints about the relationship among different attacks. %clustering

Another family of methods rely on traffic prediction to detect anomalies.
E.g., Zakroum et al.~\cite{zakroum2022monitoring} infer anomalies on network telescope traffic by predicting probing rates. They present a framework to monitor probing activities targeting network telescopes using Long Short-Term Memory deep learning networks to infer anomalous probing traffic and to
raise early threat warnings. % AI/ML prediction

% While the above techniques focus on identifying anomalies in Darknet data that are clear outliers, our approach can identify even sparse anomalies, that may be hidden in the high-dimensional Darknet data signals.

%~\cite{gupta2018identifying}

\section{Methodology}\label{sec:methodology}

\subsection{Problem Formulation}
Consider a vector time series $\vec{x}_t = \{x_t(i)\}_{i=1}^p$ modelling a collection of $p$ data streams (e.g., scanning activity
against $p$ ports). In network traffic monitoring applications, the data streams involve a highly non-stationary ``baseline'' traffic background signal $\vec{\theta}_t=\{\theta_t(i)\}_{i=1}^p,$ which could be largely unpredictable and highly variable. Considering Internet traffic specifically, this ``baseline'' traffic often includes diurnal or weekly periodic trends that can be modeled by a small number of 
common factors. 
We encode these periodic phenomena through the classic linear factor model. Namely, we make the assumption that 
$\vec{\theta}_t = B\vec{f}_t,$ where $B=(\vec{b}_1\hdots \vec{b}_k)_{p\times k}$ is a matrix of $k(\leq p)$ linearly independent columns that expresses the affected streams. 
Moreover, $\vec{f}_t=\{f_t(j)\}_{j=1}^k$ are the (non-stationary) factors, or periodic trends, that appear in these different streams. That is, the anomaly free regime can be modeled as 
\begin{equation}\label{eq:anom-free model}
 \vec{x}_t = B\vec{f}_t +\vec{\epsilon}_t,
\end{equation}
where $\vec{\epsilon}_t=\{\epsilon_t(i)\}_{i=1}^p$ is a vector time-series modeling the ``benign'' noise. The time-series $\{\vec{\epsilon}_t\}$ may and typically does have a non-trivial (long-range) 
dependence structure~\cite{willinger2001long}, but may be assumed to be stationary.

\begin{figure}
      \centering
      \includegraphics[width=0.9\linewidth]{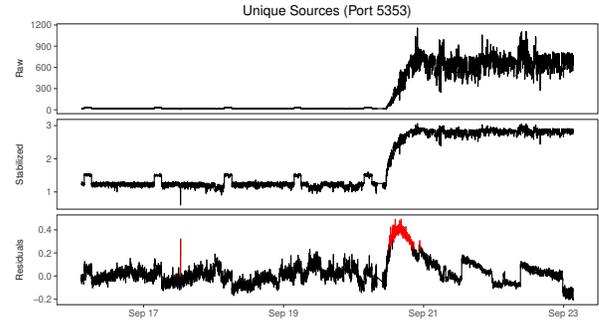}
      \caption{(a) Incoming unique source  IP traffic coming to port 5353, from 16 September, 0639 UTC to 23 September, 0719 UTC in 2016.  (b) Incoming traffic for the same port after stabilization through log-transform. Another event is now visible despite diurnal trends in traffic denoted by a spike around 17 September 14:00 UTC. (c) Residuals for the same port from the traffic model. 
      The large event is flagged by our algorithm. Our method adapts to the change of regime after the second event and does not flag the whole remaining time series as anomaly. Moreover, the periodic trend until 21 September is clearly filtered out after the application of the algorithm.}
      \label{fig:illustration}
\end{figure}

Anomalies such as aggressive network scanning may be represented (in a suitable feature space) by a mean-shift vector $\vec{u}_t=\{u_t(i)\}_{i=1}^p$, which is \textit{sparse}, i.e., it affects only a relatively small and unknown subset of streams $S_t:=\{i\in [p]: u_t(i)\neq0\}.$ Thus, in the anomalous regime, one observes
\begin{equation}\label{eq:anom model}
    \vec{x}_t = B\vec{f}_t +\vec{u}_t+\vec{\epsilon}_t.
\end{equation}
The general problems of interest are twofold:
\begin{itemize}
    \item[(i)]\textit{Detection}: Timely \textit{detection} of the presence of the anomalous component $\vec{u}_t.$
    \item[(ii)]\textit{Identification}: The estimation of the sets $S_t$ of the streams containing anomalies.
\end{itemize}

% For model identifiability, we will assume that both $\mu_t$ and $\epsilon_t$ are both mean-zero random processes. The overall level of traffic is instead modeled with a fixed, additive vector component $\nu=\{\nu(i)\}_{i=1}^p$ that does not vary over time.

\subsection{An Algorithm for Joint Detection and Identification}\label{subsec: algorithm}
We provide a high level summary of the algorithm we implement in our analysis. We examine high-dimensional data 
sourced from Darknet traffic observations and our goal is the joint detection and identification of any nefarious activity that might be present in the dataset. 

In order to achieve our goals, we utilize sequential PCA techniques
(i.e., incremental PCA or \emph{iPCA})
attempting to estimate the factor subspace spanned by the common periodic trends. Once an estimate of this space is obtained, we project the new vector of data (one observation in each port signal) on this subspace and retain the residuals of the data minus their projections. Using an individual threshold for each stream, pertaining to the stream's marginal variance and a control 
limit chosen by the operator (see Section~\ref{subsec:operator guide} for guidance in system tuning), observations are flagged as anomalies individually.

\begin{algorithm}[h]
\caption{Online identification of sparse anomalies}\label{alg:identification}
\label{alg:method}
\begin{algorithmic}
\Require smoothing parameters $\lambda,\lambda_\mu,\lambda_\sigma$;\\ %$\lambda_\nu$,
control limit $L$; effective subspace dimension $k$;\\
initial subspace estimates $\widehat B$; initial mean estimates $\widehat{\vec{\nu}}_x,\widehat{\vec{\nu}}_r$\\
initial residual marginal variance estimates $\widehat{\vec{\sigma}}^2_r$,\\
iPCA forget factor $\eta$,  robust estimation guard $R$.
% \Ensure $y = x^n$
% \State $y \gets 1$
% \State $X \gets x$
% \State $N \gets n$
\For{new data $\vec{x}_t=(x_t(j))_{j=1}^p\in\mathrm{R}^p$}
\For{$j\not\in \widehat{S}_{t-1}$}
    \State $\widehat \nu_x(j) \gets (1-\lambda)\widehat\nu_x(j) +\lambda x_t(j)$ \Comment{Update mean estimates}
\EndFor
% \State {\tt \# Update subspace}:
% \State {\tt \# Projection step}:
\State $\vec{r}_t \gets (I - \widehat{B}(\widehat{B}^{\top}\widehat{B})^{-1}\widehat{B}^{\top})(\vec{x}_t-\widehat{\vec{\nu}}_x)$ \Comment{Projection step}
\State $\widehat{B} \gets \mathrm{iPCA}(\vec{x}_t, \widehat{B}, \widehat{\vec{\nu}}_x, k)$ \Comment{Update subspace}
\For{$j$:\ $|r_t(j)|<R\cdot\widehat \sigma_r(j)$}
    \State $\widehat\nu_r(j) \gets (1-\lambda_\mu)\widehat\nu_r(j) +\lambda_\mu r_t(j)$
\EndFor
\For{$j$:\ $|r_t(j)-\widehat \nu_r(j)|<R\cdot\widehat \sigma_r(j)$}
    % \State {\tt \# Update residual variances:}
    \State $\widehat\sigma^2_r(j) \gets (1-\lambda_\sigma)\widehat\sigma^2_r(j) + \lambda_\sigma (r_t(j)-{\widehat\nu_r(j)})^2 $ \\\Comment{Update residual variances}
\EndFor
\State $\widehat{S_t} \gets \{j\;|\;|r_t(j)-\widehat \nu_r(j)| > L\times\widehat\sigma(j)\}$ 
\If{$\widehat{S_t} \neq\; $\O }
    \State Raise alert on $\widehat{S_t}$
%\ElsIf{$\widehat{S_t} =\; $\O }
%    % \State $\widehat\nu \gets (1-\lambda_\nu)\widehat\nu + \lambda_\nu x_t $ \Comment{Update mean}
%    \State $\widehat\sigma^2_r \gets (1-\lambda_\sigma)\widehat\sigma^2_r + \lambda_\sigma r^2_t $ \Comment{Update res.\ vars}
%    \State $\widehat{B} \gets \mathrm{iPCA}(x_t, \widehat{B}, \widehat\mu_x, \lambda_\Sigma, k)$ \Comment{Update subspace}
\EndIf
\EndFor
\end{algorithmic}
\end{algorithm}

The algorithm can be broken down in five phases. 
\begin{enumerate}
    \item \textbf{Input.} %We start by providing a list of tuning parameters, or parameters that the operator has a choice on. 
    Tuning parameters include the length of the warm-up period ($n_0$) and smoothing parameters for some Exponentially Weighted Moving Average (EWMA) steps~\cite{10.2307/1269835}; $\lambda,\lambda_{\mu},\lambda_{\sigma}$ denote the memory parameters for the EWMA of the mean of the data $\vec{x}_t$, the mean of the residuals $\vec{r}_t$ and the marginal variance of the residuals $(\widehat{\vec{\sigma}}_r^2)$. Additionally, the memory parameter of the {\em incremental PCA (iPCA)} (namely, $\eta$) can be chosen by the user. Moreover, one can choose the control limit $L$, used in the alert phase, the variance cut-off percentage $\pi$ and the robust guard estimation (REG) parameter $R$, controlling the sequential update of the EWMAs.         
    \item \textbf{Initialization.} The ``warm-up'' dataset, undergoes a batch PCA. Using the percentage $\pi$
    that denotes the fraction of variance ``explained", one can determine the effective dimension $k$ of the factor subspace to be estimated, which dictates the trends that would be filtered out.
    The mean of the data $(\widehat{\vec{\nu}}_x)$, the mean $(\widehat{\vec{\nu}}_r)$ and the marginal variance $(\widehat{\vec{\sigma}}_r^2)$ of the residuals are initialized based on the training data. 
    \item \textbf{Sequential Update.} A new vector of data $(\vec{x}_t)$ is observed and passed through the algorithm. We update the estimated factor subspace $\widehat B$ using incremental PCA~\cite{6483308}. The vector $\widehat{\vec{\nu}}_x$ is updated using EWMA. % with parameter $\lambda.$
    \item \textbf{Residuals.} Depending on whether there is an ongoing anomaly, we center $\vec{x}_t$. We project the centered data onto the orthogonal complement of $\widehat B$ and obtain the residual $(\vec{r}_t)$. If any coordinate of
    $\vec{r}_t$ is smaller in magnitude than $R$ times the marginal variance, then the appropriate elements of
    both $\widehat{\vec{\nu}}_r$ and $\widehat{\vec{\sigma}}_r^2$ are updated via EWMA. %'s with their respective memory parameter.  
    \item \textbf{Alerts.} Finally, if any coordinate in the centered residuals exceeds
    in magnitude the corresponding element in
    $L \times\widehat{\vec{\sigma}}_r$,  an alert is raised. 
\end{enumerate}

There are two important points that we would like to make regarding the effectiveness of the proposed Algorithm.
\begin{itemize}
    \item First, non-stationarity in the form of e.g.\ diurnal cycles in the data is absorbed in the projection step, and not with a time-series model.
    
    % Consider, for example, the scenario where a sudden global surge doubling traffic between period $(t-1)$ and $t$. Then at time $t$ we would observe
    % \begin{equation}\label{eq:global_surge}
    %     x_{t} = 2Bf_{t-1} + \mu_{t} +\epsilon_{t}.
    % \end{equation}
    % In this case the projection step, if we have accurate estimates of subspace $B$, would eliminate the baseline component $\theta_{t}=2Bf_{t-1}$.
    % In essence, we are using cross-sectional corrections to deal with serial variations that are unamenable to time-series methods.
    
    \item The sparse anomalous signals remain largely unaffected by the projection step. This can be quantified with the help of an
    incoherence condition (see Proposition~\ref{p:prop1} and its proof in Section \ref{appendix}) that is generally satisfied in network traffic monitoring problems. This then leaves us with the residuals approximating the signals
    \begin{equation} \label{eq:residuals}
        \vec{r}_{t} \approx \vec{u}_{t} + \vec{\epsilon}_t,
    \end{equation}
    and enables us to detect and locate sparse anomalies. % with good accuracy.
\end{itemize}

Moreover, the sequentially and non-parametrically updated estimate of the variance $\widehat{\vec{\sigma}}_r^2$ on Step (4) of the algorithm can be shown to be a consistent estimator in the special case that our residuals $\{\vec{r}_t\}_{t=1}^{\infty}$ form a Gaussian time series. (More general cases could be examined as future work.) 
Indeed, let $\{{r}_t(j),\ t\in\mathbb{Z},\ j\in [p]:=\{1,\dots,p\}\}$ be a zero-mean stationary 
Gaussian time series, with some correlation structure
\begin{equation}\label{eq:correlation}
    \rho_t(j) = {\rm Cor}({r}_t(j),{r}_0(j)) = {\rm Cor}({r}_{t+h}(j),{r}_h(j)),
\end{equation} 
for $\ t\in\mathbb{Z},\ j\in[p],\ \forall h\in\mathbb{Z}$. We propose the estimator   
\begin{equation}\label{e:var_estim}
    \widehat{{\sigma}}_t^{2}(j) = (1-\lambda_\sigma)\widehat{{\sigma}}_{t-1}^2(j) + \lambda_\sigma {r}_t^2(j),
\end{equation}
in order to estimate the unknown variance of $\vec{r}_t$ non-parametrically. Namely for the estimation of this variance, we implement an EWMA on the squares of the zero-mean stationary series $\{\vec{r}_t,\ t\in\mathbb{Z}\}$, which the following Proposition~\ref{prop:consistency} proves consistent.
\begin{prop}\label{prop:consistency}
Let $\{\widehat\sigma_t^2,\ t\in\mathbb{Z}\}$ be defined as in \eqref{e:var_estim} and $\{\rho_t, t\in\mathbb{Z}\}$ as in \eqref{eq:correlation} (we have dropped the index $j$ to keep the notation uncluttered). If $\sum_{t=-\infty}^{\infty}\rho_t^2<\infty$, then the estimator $\widehat \sigma_t^2$ is consistent.
\end{prop}
\begin{rem}\label{rem:variance-ewma} Note that at first glance the EWMA estimator for the variance of the residuals in \eqref{e:var_estim} seems inconsistent with the one in Algorithm~\ref{alg:identification}, due to the lack of the centering term $\widehat \nu_r(j)$. Recall that our method firstly projects the original data in the orthogonal complement of the estimated subspace $\widehat B$, obtaining the residuals $\widehat{\vec{r}}_t$, which are then centered. The centered residuals correspond to the residuals that we utilize in the proof and setting of Proposition~\ref{prop:consistency}, i.e., we assume they are a zero-mean Gaussian stationary time series with square summable correlation structure. Thus, the update of $\widehat \sigma_r^2(j)$ in \eqref{e:var_estim} does not require the existence of a centering factor.  
\end{rem}

We now proceed to elaborate on the incoherence condition, which ensures that the sparse anomalies are largely retained by 
the projection step.

\subsection{Incoherence conditions}

Disentangling sparse errors from low rank matrices has been a topic of extensive research since the last decade. While there are seminal works on recovering sparse contamination from a low-rank matrix \textit{exactly}, they all resort to solving a convex program or a program with convex regularization terms \cite{candes2012exact, candes2011robust, xu2010robust}. The form that most closely resembles the set-up of our problem is Zhou et.al. \cite{zhou2010stable}, where the model has both a Gaussian error term and a sparse corruption; see also Xu et.al.\cite{xu2010robust}, Candes et.al. \cite{candes2011robust} and references therein. 

This line of literature often criticizes the classical PCA in the face of outliers and data corruptions. However, to the best of our knowledge, there has been no analysis 
documenting just how or when the classical PCA becomes brittle in recovering sparse errors when the training period itself is contaminated. Our empirical and theoretical findings
point to the contrary.

Secondly, rarely can such convex optimization-based methods be made ``online". Our approach here is to separate sparse errors from streaming
data rather than stacked vectors of observations (i.e., measurement matrices) collected over long stretches of time. The low-rank matrix completion line of work faces challenges when data comes in streams rather than in batches, because the matrix nuclear norm often used in such problems closely couples all data points. In a notable effort by Feng, Xu, and Yan \cite{feng2013online} an iterative algorithm was proposed to solve the optimization problem with streaming data. However, implementing the algorithm requires extensive calibration and tuning which makes it impractical for non-stationary data streams with ever-shifting structures.

We now make formal the following statement: under the suitable structural assumption on the factor loading matrix $B$ stated below, PCA is still a resilient tool for recovery of the low-rank component and the sparse errors on observations $\vec{x}_t$. Our results provide theoretical guarantees on the error in recovering the locations and magnitudes of sparse anomalies.% in step two of our algorithm.

We analyze the faithfulness of residuals in \eqref{eq:residuals} in recovering sparse anomalies under the following assumptions.
\begin{equation}%[$\lambda$-min]
\label{A1}
\lambda_{\mathrm{min}}(B^\top B)\ge \phi(p), \text{ for some function } \phi \text{ of } p.
\end{equation}
Entries of $B$ are bounded by a constant $C$, uniformly in $p$:
\begin{equation}%[Uniform boundedness]
\label{A2}
|B(i,j)| \le C, \quad \forall i\in [p]:=\{1,\dots,p\}, j \in [k]
\end{equation}

Conditions \eqref{A1} \& \eqref{A2} are closely linked to the so-called incoherence conditions in high-dimensional statistics~\cite{candes2012exact}.
Notice that for \eqref{A1} and \eqref{A2} to hold simultaneously, we need $\phi(p) \le Cp$. These conditions are important for the
theoretical guarantees in the high-dimensional asymptotics regime where $p\to\infty$ but in practice they are quite mild and natural.

Condition \eqref{A1} ensures that the background signal in the factor model $B \vec{f}_t$ has enough energy or a ``ground-clearance" relative to the dimension. 
In the context of Internet traffic monitoring, this is not a restrictive condition
since otherwise the background signal can be modeled with a lower-dimensional factor model with smaller value of $k$.   The second assumption \eqref{A2}, taken together 
with \eqref{A1} ensures that the columns of $B$ are not sparse and consequently the background traffic is not concentrated on one or a few ports. This is also natural.
Indeed, if the background was limited to a sparse subset of ports, then they would always behave differently than the majority of the ports and one can simply analyze these ports separately
using a lower-dimensional model with smaller value of $p$, where we no longer have sparse background.

Let now $\Sigma = B B^\top$ and $\widehat \Sigma$ be an estimate of $\Sigma$ obtained for example, by
performing iPCA and taking the top-$k$ principal components. (Alternatively, one simply take 
$\widehat \Sigma:=n^{-1}\sum_{t=1}^n \vec{x}_t \vec{x}_t^\top$, for a window of $n$ past observations.)

%and focusing separately on the 
%and should be analyzed separately, since they would always appear as `sparse anomalies'  
%there were 
%significant variability in `normal' traffic pattern that is concentrated on one or a few ports addresses, then there would be little hope in trying to tell apart anomalies from usual 
%traffic in these IPs. After all, these IPs will behave differently from nodes on the network, and borrowing information from others does little to help identify abnormal behavior.

% These assumptions are realistic for the form of Internet traffic data at hand. Due to random hashing and aggregation, each IP hash bin are roughly similar in volume and variance; 
%the internet traffic also displays significant correlation across the nodes.

% Denote $\lambda_j$, $j\in[p]$ the eigenvalues of $\Sigma$, decreasing in magnitude; $\widehat b_0(j)$ be the corresponding eigenvectors, $\widehat{B_0}=[\widehat b_0(1),\dots,\widehat b_0(k)]$.

% Let $\mathcal P_{\widehat{B_0}}$, $\mathcal P_{\widehat{B_0}^\perp}$ be the projection onto the column space of $\widehat{B_0}$ and its orthogonal complement. We will estimate $\epsilon+u$ by the residual terms from the projection onto column space of $\widehat{B_0}$,

%[Need to change first result to $l_\infty$-norm]
\begin{prop}[Resilience of PCA]\label{p:prop1} Assume (\ref{A1}) and (\ref{A2}) hold. Then, 
for any $k$ and $p$, and for each coordinate $i\in[p]$,
\begin{multline}\label{e:six}
\E(r_t(i) - (\epsilon_t(i) + u_t(i)))^2  \le\\  
\Bigg[\left(\E\|\widehat\Sigma-\Sigma\|^2\right)^{1/2}\frac{2\sqrt{kp}}{\phi(p)}\left(c_fCk+1+\sqrt{\frac{\mbox{tr}(\Sigma_u)}{p}}\right)\\ + \frac{\sqrt{k}C}{\phi(p)}\left(1+\sqrt{kp}C\|\Sigma_u\|^{1/2}\right)\Bigg]^2
\end{multline}
where $\Sigma_{{u}} = \E [ \vec{u}_t\vec{u}_t^\top]$ and $\|\vec{f}_t\|\le c_f\sqrt{k}$.%, with $\vec{f}_t$ defined as in \eqref{eq:anom model}. 
% \shorten{Moreover, in $\ell_2$-norm, we have 
% \begin{multline}
%  \E\|\vec{r}_t-(\vec{\epsilon} + \vec{u})\|^2 \le \\
%  \Bigg[(\E\|\widehat\Sigma-\Sigma\|^2)^{1/2}\frac{2\sqrt{pk}}{\phi(p)}\left(c_fCk+1+\sqrt{\frac{\mbox{tr}(\Sigma_u)}{p}}\right)\\ + \sqrt{k}+\sqrt{\min\{k\|\Sigma_u\|,\mbox{tr}(\Sigma_u)\}}\Bigg]^2
% \end{multline}}
\end{prop}

This result, under the conditions of Proposition \ref{prop:bound on Sigma hat minus Sigma}, shows that if $p^{5/2} = o(\phi(p))$, with $k$ and $\Sigma_{{u}}$ bounded, the upper 
bound in \eqref{e:six} converges to zero, as the dimension $p$ grows. \textbf{That is, the anomaly signal 
$\vec{u}_t$ ``passes through'' to the residuals $\vec{r}_t$ and the approximation in \eqref{eq:residuals} 
can be quantified.} 
% As an illustration, for dimension $p=128$, suppose $k=5$, $C = 1$, with all entries in $B$ equal to $C$ in absolute value. Then $\phi(p) = C^2p = 128$. Let $\vec{u}$ be $s$-sparse, where $s=1$, i.e., at any time a random coordinate is experiencing anomalies with magnitude {\color{blue} $\eta$ has been used as the forget factor} $\eta=10$, then $\mbox{tr}(\Sigma_{\vec{u}}) = s\eta^2 = 100$, and $\|\Sigma_{\vec{u}}\|=s^2\eta^2/p = 1$. On each coordinate error is bounded by $\E(r(i) - \epsilon(i) - u(i))^2 \le 0.83$, while the expected error in $\ell_2$ norm is bounded by $\mathrm{E}\|\vec{r} - \vec{\epsilon} - \vec{u}\|^2 \le 9.55$. When dimensionality increases $p = 256$, keeping the error structure the same, the two error bounds becomes $0.13$ and $7.03$; for even larger dimensions, $p=1000$, the error bounds are $0.008$ and $5.49$. 
Indeed, assuming $k$ is fixed for a moment, \eqref{e:six} entails that the point-wise bound on mean-squared difference between
the unobserved anomaly-plus-noise signal $u_t(i) + \epsilon_t(i)$
and the residuals $r_t(i)$ obtained from our algorithm is 
${\cal O}(\E(\|\widehat \Sigma - \Sigma\|^2) p/\phi^2(p))$.  
This means that in practice, anomalies $u_t(i)$ of magnitudes exceeding
$p^{5/2}/\phi(p)$ will be present in the residuals $r_t(i)$.  As seen in the lower bound in 
\eqref{e:amplitude}
this is a rather mild restriction.  Thus, in view of Theorem \ref{thm: support recovery}, 
provided $\phi(p)\gg \sqrt{p/\log(p)}$, the theoretically optimal exact identification of all 
sparse anomalies is unaffected by the background signal $B\vec{f}_t$.  

In practice, however, the key caveat is the
accurate estimation of $\Sigma$, which can be challenging if the noise $\vec{\epsilon}_t$ and/or the factor signals $\vec{f}_t$
are long-range dependent. In Proposition \ref{prop:bound on Sigma hat minus Sigma}
we establish upper bounds on $\E [\|\widehat \Sigma - \Sigma\|^2]$ via the Hurst long-range dependence parameter of the 
time-series $\{\epsilon_t(i)\},$ conditions on $\{f_t\}$, and the length of the training window $n$, while the time-series $\{u_t(i)\}$ is also assumed long-range dependent.  It shows that even in the presence of long-range dependence, provided the training window and the memory of iPCA are sufficiently large, the background signal can be effectively filtered out without affecting the theoretical boundary for exact support recovery discussed in the next section.

%
%Proposition 1 serves as a guide on how well the residuals from classical PCA recovers the sparse anomalies. With 
%a small recovery error, we may, with good confidence, use residuals that are easily estimated from data to pinpoint
%the location of the anomaly.

\subsection{Statistical Limits in Sparse Anomaly Identification}
%\subsection{Dealing with trends: sequential PCA}
% {\color{red} 1) We need to explain better that the purpose of this section is that we follow theory that
% tells us that thresholding procedures work for the detection and identification tasks at hand;
% 2) Do we meet the assumptions of these theorems? }
As argued in the previous section using for example PCA-based filtering methods, 
one can remove complex non-sparse spatio-temporal background/trend signal without affecting significant {\em sparse} 
anomalies. Therefore, the sparse anomaly detection and identification problems can be addressed transparently in the context of
the signal-plus-noise model. Namely, assume that we have a way to filter out the background traffic $B\vec{f}$ induced by multivariate non-sparse trends\footnote{We drop the time subscript to keep the notation uncluttered.} in \eqref{eq:anom model}. Thus, we suppose 
that we directly observe the term:
$$
\vec{x}=\vec{\epsilon}+\vec{u}, 
$$
where $\vec{u}$ is a sparse vector with $s\ll p$ non-zero entries. 
In this context, we have two types of problems:
\begin{itemize}
    \item {\bf Detection problem.} Test the hypotheses
    \begin{equation}\label{eq:hypothesis test}
        {\cal H}_0\, :\, \vec{u}=0 \mbox{ (no anomalies)\ \ vs \ \ }{\cal H}_1\, :\, \vec{u}\not =0.      
    \end{equation}
    \item {\bf Identification problem.} Estimate the {\em sparse support set} 
    \begin{equation}\label{eq:support recovery}
     S= S_p :=\{ i\in [p]\, :\, u(i)\not = 0\}   
    \end{equation}
    of the locations of the anomalies in $\vec{u}$.%, where $[p]:= \{1,\cdots,p\}$.
\end{itemize}
Starting from the seminal works of \cite{ingster:1998,donoho:jin:2004}
the fundamental statistical limit of the {\em detection problem} has 
been studied extensively under the so-called high-dimensional asymptotic regime 
$p\to\infty$ (see, e.g., Theorem 3.1 in the recent monograph \cite{gao:stoev:2021}
and the references therein).  
The identification problem and more precisely the exact recovery of the 
support $S_p$ have only been addressed recently \cite{butucea:ndaoud:stepanova:tsybakov:2018,10.3150/20-BEJ1197}.
To illustrate, suppose that the errors have standard Gaussian distributions
$\epsilon_t(i)\sim {\cal N}(0,1)$ and are independent in $i$ (more general distributional assumptions and 
results on {\em dependent} errors have been recently developed in \cite{gao:stoev:2021}). 
Consider the following parameterization of the anomalous signal support size and amplitude as a function 
of the dimension $p$:
\begin{itemize}
\item {\bf Support sparsity.} For some $\beta\in (0,1]$,
\begin{equation}\label{e:support}
|S|\asymp p^{1-\beta},\ \ \mbox{ as }p\to\infty.
\end{equation}
The larger the $\beta$, the sparser the support.
\item {\bf Signal amplitude.} The non-zero signal amplitude satisfies
\begin{equation}\label{e:amplitude}
 \sqrt{2 \underline r \log(p)} \le u(i) \le \sqrt{2 \overline r \log(p)},\ \ \mbox{ for all $i\in S_p$.}
\end{equation}
\end{itemize}
%We allow the errors to be a sequence of weakly
%dependent stationary random variables, specifically,
%\begin{equation*}
%    \epsilon\sim N(0,\Sigma_{\epsilon}),
%\end{equation*}
%where $\Sigma_{\epsilon}$ has unit diagonal entries, and $|\Sigma_{\epsilon}(i, j)| \leq r_{|i-j|}$
%for $i\neq j$ where
%\begin{equation}\label{eq: cond on variance}
%    r_p\log p \to 0.
%\end{equation}
%Let $\beta\in (0,1]$ be the sparsity parameter and assume that $S_t$ denotes the support set of $\mu_t.$ 
%We assume the parametrization $|S_p|=:s \asymp p^{1-\beta}$; the larger the $\beta$, the lower the order of non-zero signal entries $s$. 
We have the following phase-transition result (see, e.g., Theorem 3.2 in \cite{gao:stoev:2021}).
\begin{thm}[Exact support recovery]\label{thm: support recovery} 
Let the signal $\vec{u}$ satisfy the sparsity and amplitude parameterization as in \eqref{e:support} 
and \eqref{e:amplitude}, and let
$$
g(\beta) =(1+ \sqrt{1-\beta})^2,\ \beta\in [0,1].
$$
\begin{enumerate}
    \item If $\overline {r}<  g(\beta)$, then for any signal support estimator $\widehat S_p$, we have 
    \[\lim_{p\to\infty}\mathbb{P}\left[\widehat S_p = S_p\right]=0.\]
    \item If $g(\beta) < \underline{r}$, then for the thresholding support estimator 
    $\widehat S := \{ j\in [p]\, :\, x(j)> t_p\}$ with $\P[\epsilon(i)> t_p]\sim \alpha(p)/p$, with $\alpha(p)\to 0$ such that 
    $\alpha(p)p^\delta \to \infty, \forall \delta>0$, we have
    \[\lim_{p\to\infty}\mathbb{P}\left[\widehat S_p = S_p\right]=1.\]
\end{enumerate}
\end{thm}
The above result establishes the statistical limits of the anomaly identification problem known as the {\em exact support recovery} problem
as $p\to\infty$.  It shows that for $\beta$-sparse signals with amplitudes below the boundary $g(\beta)$, there are no estimators that can fully
recover the support.  At the same time if the amplitude is above that boundary, suitably calibrated thresholding procedures are optimal and
can recover the support exactly with probability converging to $1$ as $p\to\infty$.  Based on the concentration of maxima phenomenon, 
this phase transition phenomenon was shown to hold even for strongly dependent errors $\epsilon(i)$'s for the broad class of 
thresholding procedures (Theorem 4.2 in \cite{gao:stoev:2021}).
This recently developed theory shows that ``simple'' thresholding procedures are optimal in identifying sparse anomalies.  
It also shows the fundamental limits for the signal amplitudes as a function of sparsity where signals with insufficient 
amplitude cannot be fully identified in high dimension by {\em any} procedure. These results show that 
our thresholding-based algorithms for sparse anomaly {\em identification} are essentially optimal in high-dimensions
provided the background signal can be successfully filtered out.  More on the sub-optimality of certain popular {\em detection} procedures can be found in Section \ref{s:Q-stat}.

\section{Performance Evaluation}

Next, we provide experimental assessments of our methodology using both \emph{synthetic}
and \emph{real-world} traffic traces. 
We utilize two criteria to provide answers to the two-fold objective of our analysis. For the detection problem, we are only interested on whether any observation at a fixed time point is flagged as an anomaly. If so, the whole vector is treated as an anomalous vector 
and any F1-scores or ROC curves are ``global" (i.e., for any port) and time-wise. On the other hand, for the identification problem, we are also interested in correctly flagging the positions of the anomalous data. Namely, looking at a vector of observations, representing the port space at a specific time point, we want to correctly find which individual ports/streams include the anomalous activity. We report a percentage of correct anomaly identifications per time unit.

\subsection{Performance and Calibration using Synthetic Data}
\subsubsection{Linear Factor Model}
% In this section we demonstrate the validity of our method through synthetic data
% sourced via the linear factor model of Eq.~\ref{eq:anom model}.

We simulate ``baseline'' traffic data of 5 weeks via the linear factor model of Eq.~\ref{eq:anom model}, 
obtaining one observation every 2 minutes. This leads to a total of 25200 observations. 
The baseline traffic is created as a fractional Gaussian noise (fGn) with Hurst parameter $H=0.9$ and variance set to 1, 
leading to a long range dependent sequence. 
%An extensive analysis not shown here yields long-range
% dependence exponents around $H=0.9$. 

Together with the fGn, five extra sinusoidal curves---representing traffic trends---are blended in to create the final time-series. These sinusoidal curves reflect the diurnal and weekly trends that show up in real world Internet traffic. 
The first two inserted trends are daily, the third one is weekly and the last two of them have a period of 6 and 4.8 hours respectively. These trends do not affect every port in the same way; they all have a random offset and only influence the ports determined by our factor matrix $B$. In this matrix every row represents a port and the column specifies the trend; a presence of 1 in the position $(i,j)$ means that the $j$-th trend is inserted in the $i$-th port. The matrix $B$ is created by having all elements of the first column equal to 1, meaning that all ports are affected by the first trend. The number of ones in the rest of the columns is equal to $(1-\lfloor (j-1)/k\rfloor)*100\%$, where $j$ is the index of the trend and $k$ the total number of them. The ones are distributed randomly in the columns. 

We simulate traffic data that pertain to $N$ distinct ports, leading to a matrix of observations of dimensions $N\times 25200$ (e.g., $N=100$). Since we are interested in tuning the parameters of our model, we create 5 independent replications of this dataset. The first two weeks of observations, namely the first 10080 observations, are used as the warm-up period where the initial batch PCA is run in order to determine the number of principal components we will work with in the incremental PCA and proceed with initialization. 

The anomaly is inserted in the start of the fourth week; the magnitude is 
determined by the {\em signal to noise ratio (snr)}
and the number of anomalous observations depends on the given {\em duration}. We have three choices of snr and two choices of duration in this simulation setup, leading to six combinations. The options for snr are 2, 3 and 7, leading to an additional traffic of size 2, 3 and 7 times the empirical standard deviation of the corresponding port the anomaly is inserted into. Moreover, the duration of the anomalies we use is 1 hour for our ``short'' anomalies, or 30 observations in this specific data, and 6 hours for our ``long'' anomalies; 180 observations respectively. The anomalies are always inserted in the first 3 simulated ports.

%since the specific port location will not affect  the end results. 

\subsubsection{Demonstration of Sequential PCA}
Figure~\ref{fig:angle_plot} illustrates the performance of the incremental PCA algorithm as a function of  
its memory parameter---the reciprocal of the current time index parameter in~\cite{arora2012stochastic,onlinePCA}.  
Since our methodology depends on how well we approximate the factor 
subspace, we need to measure distance between subspaces of $\R^p$.  We do so in terms of the
largest principal angle $\angle (\widehat W,W)$ (\cite{bjorck1973numerical,yu2015useful}) 
between the known true subspace $W$ and the estimated one $\widehat W$,  defined as:
$$
\angle (\widehat W,W):= {\rm acos}(\sigma_1),\ \mbox{ where}\ \sigma_1 = \mathop{\max}_{\substack{\hat{\vec{w}} \in \widehat W,\ \vec{w}\in W\\ \|\hat{\vec{w}}\| = \|\vec{w}\|=1}} \vec{w}^\top \widehat{\vec{w}}.$$ 
As we can see in Figure~\ref{fig:angle_plot}, there is a ``trade-off'' with regards to the length of the memory and the performance. If this parameter is chosen to be too big, the angle of the two subspaces becomes the largest. This could be explained as relying too much on the latest observations, estimating a constantly changing subspace. If the memory is too long, then the impact of new observations is minimal, leading to very little adaptability of our estimator. In Figure \ref{fig:angle_plot}, for the data used here, using 10 weeks worth of observations and 10 replications of the data, the best memory parameter seems to be $10^{-5}$. Moreover, as expected a full PCA estimated subspace has smaller angle than the iPCA estimated one (see Figure~\ref{fig:angle_plot}).

\begin{figure}[t]
\centering
\includegraphics[scale=0.6]{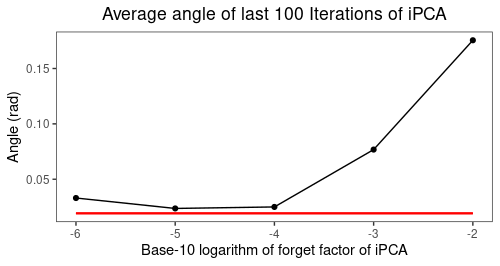}
\vspace{-7pt}
\caption{The largest principal angle between estimated and true subspace for iPCA (cf \cite{bjorck1973numerical,yu2015useful}). The horizontal line (in red) shows the 
largest principal angle between the true subspace and the batch PCA estimated one.}
\vspace{-10pt}
\label{fig:angle_plot}
\end{figure}

\subsubsection{Parameter Tuning: Operator's Guide}\label{subsec:operator guide}
Next, we proceed with hyper-parameter tuning for Algorithm~\ref{alg:method}.
Our grid search for tuning involves the following parameters and values. We have an EWMA memory parameter for the data ($\lambda$), keeping track of their mean, so that this mean can be utilized in the incremental PCA. We also have an EWMA parameter keeping track of the mean of the residuals ($\lambda_\mu$) in our algorithm. Both of these parameters are explored over the values  $10^{-2},10^{-3}$ and $10^{-4}.$ Additionally, there is an EWMA parameter for the sequential updating of the marginal variance of the residuals ($\lambda_\sigma$); the possible values we examine are $10^{-4},10^{-5}$ and $10^{-6}.$ Finally, we explore the values of control limit ($L$) that work best with the detection of anomalies in this data; we have a list of values including $0, 10^{-4},10^{-3},10^{-2}, 0.1, 0.5,1,2,3,4,5,6,7,20.$ 

The best choices for our tuning parameters are evaluated based on a combination of $F1$-score and AUC (area under curve) value and can be found on Table \ref{table:recommendations}. We start by selecting a combination of {\em snr} and {\em duration} and we find the AUC for each different combination of the rest of the tuning parameters with respect to $L$. We choose the parameters that lead to the largest AUC and then elect the control limit that corresponds to the largest $F1$-score. Our definition of $F1$-score is based on individual flagging of observations as discoveries or not; namely we look at the $15120\times100$ ``test'' observations as a vector of 1512000 observations and each one of them is categorized as true/false positive or true/false negative. 
%\textcolor{red}{Include the F1-scores and tpr/fpr in the table?}

\begin{table}[t]
\caption{Recommendations for the choice of tuning parameters
% In the table one can find the combination of magnitude and length of anomaly, with the respective choice of L and EWMA parameters. 
}\label{table:recommendations}
\centering
\begin{tabular}{rrrrrrrrrrrrr}
  \hline
 snr & duration & L & REG & ewma\_data & ewma\_mean & ewma\_var\\ 
  \hline
  2 & 1 & 5 & 4 & 0.0010 & 0.0001 & 0.00010 \\ 
  5 & 1 & 7 & 5 & 0.0010 & 0.0010 & 0.00001 \\ 
  7 & 1 & 7 & 5 & 0.0001 & 0.0100 & 0.00010 \\ 
  2 & 6 & 5 & 3 & 0.0001 & 0.0010 & 0.00010 \\ 
  5 & 6 & 7 & 5 & 0.0010 & 0.0010 & 0.00001 \\ 
  7 & 6 & 7 & 3 & 0.0001 & 0.0100 & 0.00010 \\ 
   \hline
\end{tabular}
\end{table}

\subsection{Sub-optimality of classic Chi-square statistic detection methods} \label{s:Q-stat}

{The detection problem can be understood transparently in the
signal-plus-noise setting.  Namely, thanks to Proposition \ref{p:prop1}, 
suppose that we observe a high-dimensional vector as $\vec x_t = \vec u_t + \vec\epsilon_t$, 
where $\vec u_t = (u_t(i))_{i=1}^p$ is a (possibly zero) anomaly signal. Then, the detection 
problem can be cast as a (multiple) testing problem 
$$
{\cal H}_0 \, :\, (u_t(i))_{i=1}^p = \vec{0}\ \mbox{ vs }\ {\cal H}_a\, :\, u_t(i)\not=0,\ \mbox{ for some }i\in [p].
$$
}
%Recall the detection problem and the associated hypothesis test demonstrated in
%\eqref{eq:hypothesis test}. 
Under the assumption that $\vec{\epsilon}_t \sim N(0,\Sigma_{p\times p}),$ one 
popular statistic for detecting an anomaly (on any port), i.e., testing ${\cal H}_0$ is:
%for the presence of anomaly  the above that come in mind is to use the following 
%$\chi^2$ statistic for the detection problem.  
\[Q = \vec{x}^{\top}\Sigma^{-1}\vec{x}\sim\chi_p^2, \ {\rm under\ } \mathcal{H}_0.\]
%{\color{red} Capital $X$ or lower-case $x$? Pls fix my paragraph above accordingly.} 
Note that the statistic is a function of time, but from now on, for simplicity 
we shall suppress the dependence on $t$.  Using this statistic we have the following rule: 
We reject the $H_0$ at level $\alpha>0,$ if $Q>\chi_{p,1-\alpha}^2.$

Coming back to our algorithm, we have a control limit $L,$ which we can use to raise alerts in our alert matrix. Namely, we raise an alert on port $i$ at time $t$, if the $i$-th residual at time $t$ is bigger than $L$ times the marginal variance of these residuals. 
Note that at every time-point, Algorithm~\ref{alg:method} (i.e., the iPCA-based method) gives 
us individual alerts for anomalous ports, while the Q-statistic only 
flags presence/absence of an anomaly over the port dimension. To perform a fair comparison of the two methods, we 
use the following definition of true positives and false negatives for Algorithm~\ref{alg:method}. We declare an alert 
at time $t$ a true positive, if there is at least one anomaly in the port space at time $t$ and Algorithm~\ref{alg:method} raises 
at least one alarm at time $t$; not necessarily at a port where the anomaly is taking place. 
A false negative takes place if an anomaly is present at the port space, but no alerts are being 
raised across the port space at time $t$.

To compare the two methods, we start with an initial warm-up period which we apply a batch PCA approach to. 
Then, we use the estimate of the variance for the Q-statistic method, while we also obtain the number of principal
components to keep track of during the iPCA, using the number of principal components that explain 90\% of the 
variability of this warm sample. We choose the tuning parameters for iPCA based on the recommendations provided in 
Section \ref{subsec:operator guide}.

Figure \ref{fig: long duration} illustrates the performance of the two methods as the number of ports increases, 
while the sparsity of the anomalies in the simulated traffic remains 
the same. The {\em sparsity parameter} $\beta$ (see~\cite{gao:stoev:2021}) 
is used to control the sparsity of the anomalies; for a number of ports $p$ we 
insert anomalies in $\lfloor p^{1-\beta}\rfloor$ ports, where $\beta =3/4$. As is evident from the plots, Algorithm~\ref{alg:method} 
significantly outperforms the Q-statistic as the port dimension grows.  

This sub-optimality phenomenon is well-understood in the high-dimensional inference literature 
(see \cite{fan1996test} and Theorem 3.1 \cite{gao:stoev:2021}).  Namely, as $p\to\infty$, the $Q-$statistic 
based detection method will have vanishing power in detecting sparse anomalies in comparison with the optimal 
methods such as Tukey's higher-criticism statistic.  As shown in Theorem 3.1 \cite{gao:stoev:2021}, the 
thresholding methods like Algorithm~\ref{alg:method} are also optimal in the very sparse regime $\beta \in [3/4,1]$ in the sense that
they can discover anomalous signals with magnitudes down to the theoretically possible {\em detection boundary} 
(cf \cite{ingster:1998}). \textbf{This explains the growing superiority of our method as the port dimension 
increases, as depicted in Figure~\ref{fig: long duration}.}

\begin{figure}
    \centering
    \includegraphics[scale=0.33]{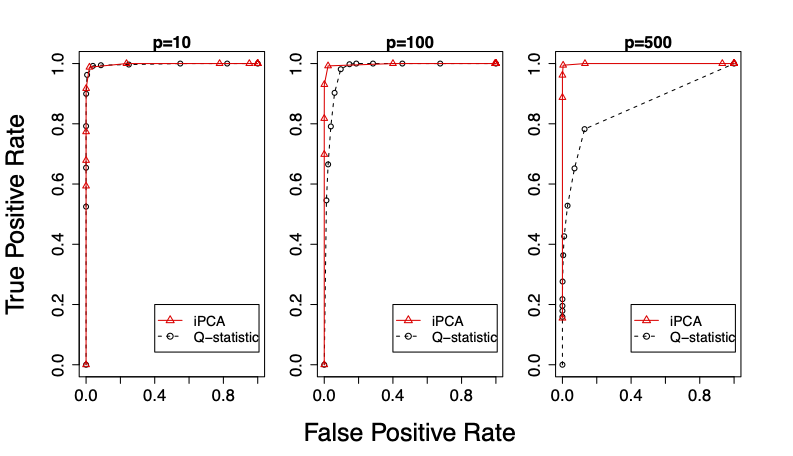}
    \vspace{-10pt}
    \caption{ROC curves of Algorithm~\ref{alg:method} (iPCA-based) and Q-statistic for long duration anomaly (6 hour) of low magnitude (snr=2) in synthetic data.}
    \label{fig: long duration}
    %\label{fig:my_label}
\end{figure}

At the same time, the proposed iPCA-based method also has the advantage of providing us with extra information in comparison 
to the Q-statistic. Indeed, using the individual time/space alerts, one can identify what percentage of individual true positives/false negatives exist in the time/port space. An example of this information is shown on Table \ref{table:individual tpr}, where the ``rows'' variables refer to the detection problem and ``indiv'' variables refer to the identification one. Namely, the ``rows'' variables are only concerned with temporal events, while the ``indiv'' variables take into account the spatial structure as well.
% \begin{table}[ht]
% \caption{Individual and time-wise evaluation of the iPCA method}\label{table:individual tpr}
% \centering
% \begin{tabular}{rrrrrr}
%   \hline
%  L & tpr\_rows & fpr\_rows & F1\_rows & tpr\_indiv & fpr\_indiv  \\ 
%   \hline
%  2.00 & 1.00 & 1.00 & 0.00 & 1.00 & 1.00  \\ 
%  3.00 & 1.00 & 0.92 & 0.00 & 1.00 & 0.92  \\ 
%  4.00 & 1.00 & 0.11 & 0.04 & 0.99 & 0.11  \\ 
%  5.00 & 1.00 & 0.00 & 0.65 & 0.97 & 0.00 \\ 
%  6.00 & 1.00 & 0.00 & 1.00 & 0.93 & 0.00 \\ 
%  7.00 & 1.00 & 0.00 & 1.00 & 0.87 & 0.00 \\ 
%  20.00 & 0.08 & 0.00 & 0.12 & 0.02 & 0.00 \\ 
%   \hline
% \end{tabular}
% \end{table}
\begin{table}[t]
\vspace{-10pt}
\caption{Performance of Algorithm~\ref{alg:method} on Synthetic Data.
Observe the high True Positive Rate (TPR) and low False Positive Rate (FPR)
for carefully chosen thresholds $L$.}\label{table:individual tpr}
\centering
\begin{tabular}{rrrrrr}
  \hline
 L & tpr\_rows & fpr\_rows  & tpr\_indiv & fpr\_indiv  \\ 
  \hline
 \textbf{2} & 1.00 & 1.00 & 1.00 & 1.00  \\ 
 \textbf{3} & 1.00 & 0.92 & 1.00 & 0.92  \\ 
 \textbf{4} & 1.00 & 0.11 & 0.99 & 0.11  \\ 
 \textbf{5} & 1.00 & 0.00 & 0.97 & 0.00 \\ 
 \textbf{6} & 1.00 & 0.00 & 0.93 & 0.00 \\ 
 \textbf{7} & 1.00 & 0.00 & 0.87 & 0.00 \\ 
 \textbf{20} & 0.08 & 0.00 & 0.02 & 0.00 \\ 
  \hline
\end{tabular}
\end{table}

\subsection{Application to real-world Darknet data}
Finally, we demonstrate the performance of our method on real-world data
obtained from the ORION Network Telescope~\cite{orion}.
The dataset spans the entire month of September 2016
and includes the early days of the infamous Mirai botnet~\cite{203628}.
We constructed minute-wise time series for all TCP/UDP ports present in our data
that represent the number of unique sources targeting a port at a given time.
For the analysis here, we focused on the top-50 ports based on their frequency in the
duration of a month. We implement the algorithm on the data by utilizing the calibrated tuning parameters suggested in \ref{subsec:operator guide}. We use the first 5000 observations, roughly three and a half days, as our burn-in period to initialize the algorithm. Thus, we have a matrix of $43200\times50$ Darknet data observations. 

The selected dataset includes important security incidents.
\textbf{Indeed, as Figure~\ref{fig: darknet ports} shows,
we detect the onset of scanning activities against TCP/2323 (Telnet for IoT) around September 6th
that can be attributed to the Mirai botnet}. We also detect some interesting 
ICMP-related activities (that we represent as port-0) towards the end of the month. We
are unsure of exactly what malicious acts these activities represent; however, upon payload inspection we found
them to be related with some heavy DNS-related scans (payloads with DNS queries were encapsulated
in the ICMP payloads).

In Figures \ref{fig: darknet ports} and \ref{fig:alert_matrix}, we are looking for large anomalies (snr=7) of long duration (6 hours).
Moreover, inspecting the full alert matrix in Figure~\ref{fig:alert_matrix}, we 
observe a few more alerts for anomalies that deserve further investigation. Some of them might
be false positives, although there is no definitive ``ground truth" in real-world data
and all alerts merit some further forensics analysis. The encouraging observations 
from Figures~\ref{fig: darknet ports}and~\ref{fig:alert_matrix} are that 
the incidents beknownst to us are revealed, and that
the elected hyper-parameters
avoid causing the so-called ``alert-fatigue" to the analysts. At the same time, 
the analysts could tune Algorithm~\ref{alg:method} to their preferences,
and prioritize which alerts to further investigate based on attack severity, duration,
number of concurrently affected ports, etc. All this information is readily available
from our methodology.

\begin{figure}[t]
    \centering
    \includegraphics[scale=0.43]{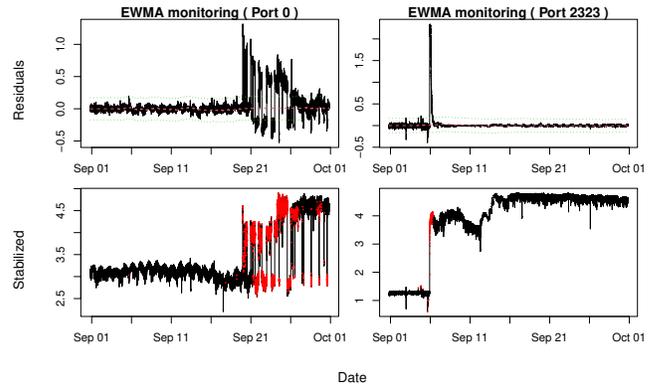}
    \vspace{-10pt}
    \caption{Residuals and detection  boundaries (top) and raw traffic (bottom) for ports 0 and 2323. Red color in bottom plots shows the detected anomalies.}
    \label{fig: darknet ports}
\end{figure}
% \begin{figure}[t]
%   \centering
%   \begin{minipage}[b]{0.22\textwidth}
%     \includegraphics[width=\textwidth,height =100pt]{Rplot09.png}
%   \end{minipage}
%   \hfill
%   \begin{minipage}[b]{0.22\textwidth}
%     \includegraphics[width=\textwidth,height =100pt]{Rplot10.png}
%   \end{minipage}
% \caption{Residuals and detection  boundaries (top) and raw traffic (bottom) for ports 0 and 2323. Red color in bottom plots shows the detected anomalies.{\color{red} Add proper times} }\label{fig: darknet ports}  
% \end{figure}

\begin{figure}[t]
\centering
\includegraphics[scale=0.42]{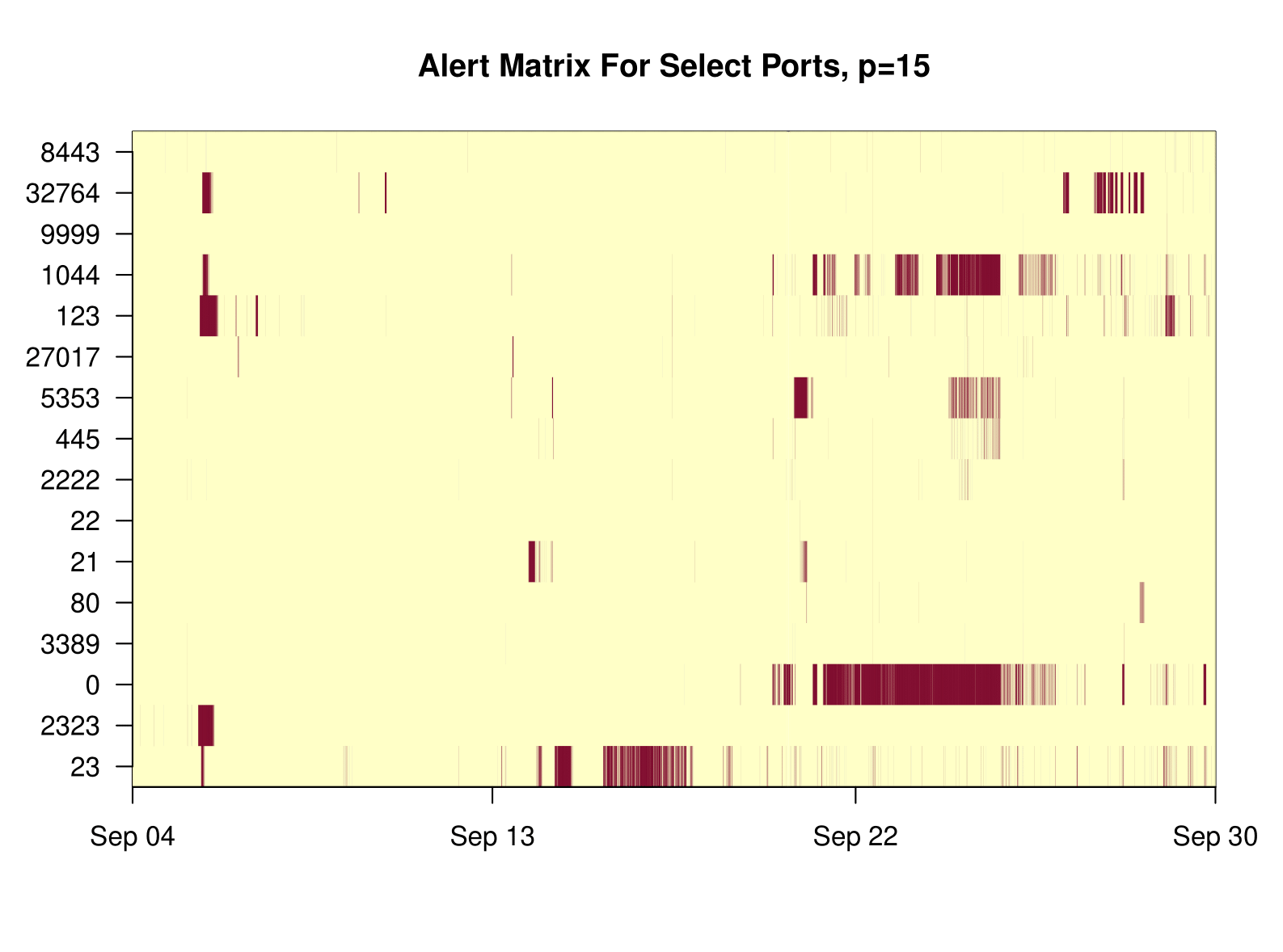}
\caption{Alert matrix for the Darknet dataset for September 2016 (Mirai onset). Red color indicates detection.}
\label{fig:alert_matrix}
\end{figure}
% \begin{figure}[t]
%   \centering
%   \begin{minipage}[b]{0.22\textwidth}
%     \includegraphics[width=\textwidth,height =100pt]{darknet_port0.png}
%   \end{minipage}
%   \hfill
%   \begin{minipage}[b]{0.22\textwidth}
%     \includegraphics[width=\textwidth,height =100pt]{darknet_porpert2323.png}
%   \end{minipage}
% \caption{Residuals and detection  boundaries (top) and raw traffic (bottom) for ports 0 and 2323. Red color in bottom plots shows the detected anomalies.{\color{red} Add proper times} }\label{fig: darknet ports}  
% \end{figure}

% \begin{figure}[t]
% \centering
% \includegraphics[scale=0.2]{alert_matrix.png}
% \caption{Alert matrix for the Darknet dataset. Red color indicates detection. 
% First 5000 observations are warm-up, so no positives are shown. {\color{red} Add proper times}}
% \label{fig:alert_matrix}
% \end{figure}
% \section{Conclusions}

\bibliographystyle{IEEEtran}
\bibliography{icc2023}

\section{Appendix}\label{appendix}
\subsection{Technical results for Section~\ref{sec:methodology}}
\subsubsection{Proof for Proposition~\ref{p:prop1}}

We introduce some further notations before we present the proof.

Let $\vec{b(j)}$ be the $j$-th column of $B$, $j\in[k]$; $\vec{b_0(j)} = \vec{b(j)} / \|\vec{b(j)}\|$ the factor loadings normalized;  $B_0 := [\vec{b_0(1)},\dots, \vec{b_0(k)}]$.  We shall denote by $\widehat B$ and $\widehat B_0$ the corresponding estimates of $B$ and $B_0$
obtained for example by the iPCA algorithm.  

For the purpose of the following theoretical results, we shall assume that
$\widehat B$ and $\widehat B_0$ are obtained using singular value decomposition of the sample covariance matrix
\begin{equation} \label{eq:Sigma-hat}
    \widehat \Sigma := \frac{1}{n} \sum_{t=1}^n \vec{x}_t \vec{x}_t^\top.
\end{equation}

\begin{proof}
Residuals $\vec{r}$ from projecting observations onto the space spanned by the first $k$ eigenvectors, are
\begin{align*}
\vec{r} &= (\mathcal I_p - \mathcal P_{\widehat{B_0}})(B\vec{f} + \vec{\epsilon} + \vec{u})\\
  &= (\mathcal I_p - \mathcal P_{\widehat{B_0}})B\vec{f} + (\mathcal I_p - \mathcal P_{\widehat{B_0}})\vec{\epsilon} + (\mathcal I_p - \mathcal P_{\widehat{B_0}})\vec{u}
\end{align*}
or equivalently, we can write 
\begin{align*}
\vec{r} -\vec{\epsilon}-\vec{u}&=  \underbrace{(\mathcal I_p - \mathcal P_{\widehat{B_0}})B\vec{f}}_{I} - \underbrace{ \mathcal P_{\widehat{B_0}}\vec{\epsilon}}_{I\!I} - \underbrace{\mathcal P_{\widehat{B_0}}\vec{u}}_{I\!I\!I}
\end{align*}

We will establish upper bounds on all three terms. The first term,
\begin{align*}
{I} &= (\mathcal I_p - \mathcal P_{{B_0}} + \mathcal P_{{B_0}} - \mathcal P_{\widehat{B_0}})B\vec{f} \\
                &= (\mathcal I_p - \mathcal P_{{B_0}})B\vec{f} + (\mathcal P_{{B_0}} - \mathcal P_{\widehat{B_0}})B\vec{f} \\
                &= -(\mathcal P_{\widehat{B_0}}-\mathcal P_{{B_0}})B\vec{f}.
\end{align*}

Using the relationship \[\|B\|\le \sqrt{p}\max_{i=1,\hdots,p}\left(\sum_{j=1}^kB_{i,j}^2\right)^{1/2}\]
and \eqref{A2} we obtain that $\|B\|\le\sqrt{pk}C$. By Corollary \ref{cl:bound} discussed below,
\begin{align}\label{term I bound}
&\E\|(\mathcal P_{\widehat{B_0}} - \mathcal P_{B_0})B\vec{f}\|
    \le  \E \Big[ \|(\mathcal P_{\widehat{B_0}} - \mathcal P_{B_0})\| \|B\| \|\vec{f}\| \Big] \nonumber\\
    &\quad\quad \le \|B\|\|\vec{f}\| \Big(\E \|\mathcal P_{\widehat{B_0}} - \mathcal P_{B_0}\|^2\Big)^{1/2}  \nonumber\\
    &\quad\quad \le \sqrt{pk}C\frac{c_f\sqrt{k} }{\lambda_{\mathrm{min}}}2\sqrt{k}\left(\mathbb{E}\|\widehat\Sigma-\Sigma\|^2\right)^{1/2}  \; \nonumber\\
    &\quad\quad \le \frac{2c_fCk\sqrt{pk}}{\phi(p)}\left(\mathbb{E}\|\widehat\Sigma-\Sigma\|^2\right)^{1/2}
\end{align}
where we used \eqref{A1} in the last inequality and which also provides a bound on each coordinate.

For the second term, $I\!I$, decomposed as $\vec{\epsilon} - \mathcal P_{B_0}\vec{\epsilon} - (\mathcal P_{\widehat{B_0}} - \mathcal P_{B_0})\vec{\epsilon}$, notice that $\mathcal P_{B_0}\vec{\epsilon} = B_0B_0^{\top}\vec{\epsilon} \stackrel{d}{\sim} B_0\vec{\epsilon}_k$ where $\E\vec{\epsilon}_k = 0$ and $\E\vec{\epsilon}_k\vec{\epsilon}_k^{\top}=I_k$. The second moment of the $i$th coordinate is bounded by
$$
\E(\mathcal P_{B_0}\vec{\epsilon})_i^2 = \E(\sum_{j=1}^k B_0(i,j)\epsilon_k(j))^2 \le \frac{kC^2}{\lambda_{\mathrm{min}}}
$$
because $|B_0(i,j)|\le C/\sqrt{\lambda_{\mathrm{min}}}$. While in $\ell_2$ norms, $\E\|\mathcal P_{B_0}\vec{\epsilon}\|^2 = \E\mathrm{tr}(B_0B_0^{\top}\vec{\epsilon}\vec{\epsilon}^{\top}B_0B_0^{\top}) = k$.
The last part is controlled similarly as in $I$.  Indeed, by the Cauchy-Schwartz inequality, 
Corollary \ref{cl:bound} and \eqref{A1}, we obtain
\begin{align*}
\Big(\E\|(\mathcal P_{\widehat{B_0}} - \mathcal P_{B_0})\vec{\epsilon}\| \Big)^2 &\le \E \|\mathcal P_{\widehat{B_0}} - \mathcal P_{B_0}\|^2 \E\|\vec{\epsilon}\|^2 \\
& \le \frac{4k\mathbb{E}\|\widehat \Sigma - \Sigma\|^2 p}{\phi(p)^2}.
\end{align*}

For the  third term, we need the von Neumann's trace inequality: for two matrices $X$ and $Y$, $|\mbox{tr}(XY)| \le \langle \vec{\sigma}(X), \vec{\sigma}(Y)\rangle$, where $\vec{\sigma}(X)$, $\vec{\sigma}(Y)$ are vectors of the singular values of $X$ and $Y$. By Holder inequality, $|\mbox{tr}(XY)| \le \|\vec{\sigma}(X)\|_1\|\vec{\sigma}(Y)\|_\infty = \mbox{tr}(X)\|Y\|$, for real symmetric positive definite matrices. 

Decompose $I\!I\!I$ as $\vec{u} - \mathcal P_{B_0}\vec{u} - (\mathcal P_{\widehat{B_0}} - \mathcal P_{B_0})\vec{u}$. The second moment of the $l$-th coordinate
\begin{align*}
\E(\mathcal P_{B_0}u)_l^2 &= \E(\sum_{i=1}^p u_i\underbrace{(\sum_{j=1}^k B_0(lj)B_0(ij))}_{a_{il}})^2 \\
        &= \E(\sum_{i=1}^p a_{il} u_i)^2 = \mbox{tr}(\vec{a}_l\vec{a}_l^{\top}\E \vec{u}\vec{u}^{\top}) \\
        &\le \mbox{tr}(\vec{a}_l\vec{a}_l^{\top})\|\Sigma_u\| \le \frac{k^2C^4p\|\Sigma_u\|}{\lambda_{\mathrm{min}}^2},
\end{align*}
where $\vec{a}_l = (a_{1l},\dots,a_{pl})^{\top}$, and the last inequality because $|a_{il}|\le kC^2/\lambda_{\mathrm{min}}$.
% In $l_2$ norm, $\E\|\mathcal P_{B_0}\vec{u}\|^2 = \mathrm{tr}(B_0B_0^{\top}\Sigma_u) \le \min\{k\|\Sigma_u\|,\mbox{tr}(\Sigma_u)\}$. 
The last part in $I\!I\!I$, again using Corollary~\ref{cl:bound} and \eqref{A1}, is bounded by
\begin{align*}
(\E\|(\mathcal P_{\widehat{B_0}} - \mathcal P_{B_0})\vec{u}\|)^2 & \le \E\|\mathcal P_{\widehat{B_0}} - \mathcal P_{B_0}\|^2 \E\mbox{tr}(\vec{u}\vec{u}^{\top}) \\ &\le \frac{4k\E\|\widehat \Sigma-\Sigma\|^2\mbox{tr}(\Sigma_u)}{\phi(p)^2}    
\end{align*}

Taken together,
\begin{align*}
    \E(\vec{r} - \vec{\epsilon} - \vec{u})_i^2 \le \left[\sqrt{\E(I_i)^2}+\sqrt{\E(I\!I_i)^2}+\sqrt{\E(I\!I\!I_i)^2}\right]^2\\ \le \Bigg [\frac{2c_fCk\sqrt{pk}}{\phi(p)}\left(\E\|\widehat\Sigma-\Sigma\|^2\right)^{1/2}  + \frac{2\sqrt{kp}}{\phi(p)}\left(\E\|\widehat\Sigma-\Sigma\|^2\right)^{1/2} + \\  \frac{\sqrt{k}C}{\phi(p)}+ \frac{kC^2\sqrt{p}\|\Sigma_u\|^{1/2}}{\phi(p)} + \frac{2\sqrt{k}\left(\E\|\widehat\Sigma-\Sigma\|^2\right)^{1/2} \sqrt{\mbox{tr}(\Sigma_u)}}{\phi(p)}\Bigg]^2 \\
    = \Bigg[\left(\E\|\widehat\Sigma-\Sigma\|^2\right)^{1/2}\frac{2\sqrt{kp}}{\phi(p)}\left(c_fCk+1+\sqrt{\frac{\mbox{tr}(\Sigma_u)}{p}}\right)\\ + \frac{\sqrt{k}C}{\phi(p)}\left(1+\sqrt{kp}C\|\Sigma_u\|^{1/2}\right)\Bigg]^2.
\end{align*}

\subsubsection{Distance between $\Sigma$ and $\hat \Sigma$}
Looking at Proposition \ref{p:prop1}, one can see that we need a bound on 
\[\E \|\widehat{\Sigma}-\Sigma\|^2.\]
Recall that the signal $u_t$ is assumed random such that 
$$
\Sigma_u := \E[u_t u_t^\top],
$$
while by \eqref{eq:anom model}, $x_t=Bf_t+u_t+\epsilon_t.$ Let also $\Sigma_\epsilon = \E[ \epsilon \epsilon ^\top]$ and $\Sigma := B^\top B$.

When we see some data contaminated with sparse anomalies, we can estimate
\begin{align*}
\widehat \Sigma  = \frac{1}{n}\sum_{t=1}^nx_tx_t^{\top} &= \Sigma + B \frac{1}{n}\sum_{t=1}^n (f_t f_t^\top -{\cal I}_k) B^\top 
+ \frac{1}{n}\sum_{t=1}^n u_t u_t^\top \\ &+ \frac{1}{n}\sum_{t=1}^n \epsilon_t \epsilon_t^\top 
 + B \frac{1}{n} \sum_{t=1}^n f_t(u_t + \epsilon_t)^\top  \\ &+  \frac{1}{n} \sum_{t=1}^n (u_t + \epsilon_t) f_t^\top B^\top\\
& =: \Sigma + B R_f B^\top + \hat\Sigma_u + \hat \Sigma_\epsilon + B R + R^\top B^\top\\ &+\frac{1}{n}\sum_{t=1}^n\epsilon_tu_t^{\top}+\frac{1}{n}\sum_{t=1}^n u_t\epsilon_t^{\top},
\end{align*}
where 
\begin{align*}
    R_f & := \frac{1}{n}\sum_{t=1}^n (f_t f_t^\top -{\cal I}_k)\\ 
    R&:= \frac{1}{n} \sum_{t=1}^n f_t(u_t + \epsilon_t)^\top\\
    \hat \Sigma_{\epsilon} &:= \frac{1}{n}\sum_{t=1}^n \epsilon_t \epsilon_t^\top \\
    \hat \Sigma_{u} &:= \frac{1}{n}\sum_{t=1}^n u_t u_t^\top. 
\end{align*}
In order to obtain upper bounds
on the squared 
expectation of the operator norm $\E[ \| \widehat \Sigma - \Sigma\|^2],$
 we will need upper bounds on $\E \|R_f\|^2, \E[ \| \widehat \Sigma_u\|^2]$, $\E[ \| \widehat \Sigma_\epsilon\|^2]$, $\E[ \|R\|^2]$ and $\E\|\frac{1}{n}\sum_{t=1}^n\epsilon_tu_t^{\top}\|^2.$

% {\bf NOTE:} The term $B R_f B^\top = \lambda_1 \times {\cal O} (\|R_f\|)$, so we will have 
% $\lambda_1^2/\lambda_k^2 =p^2/\phi(p)^2$ in the upper bound.  {\bf This is fine, though,} since 
% this term is multiplied by ${\cal O} (\|R_f\|)$, which vanishes...

% Let's quickly work-out the rest of the bounds without worrying too much.  For example, 

% \begin{itemize}
%     \item 
% [(a)] Put the Long-range dependence in $f_t$ and make assumptions on $R_f$

% \item [(b)] Assume $\epsilon_t$'s are iid, and use the ``standard'' concentration inequalities about $\|\widehat \Sigma_\epsilon\|$

% \item [(c)] Deal similarly with $u_t$.  In fact, $u_t$ can be assumed constant or handled like $\epsilon_t$.

% \end{itemize}

We start with $\E\|\hat\Sigma_\epsilon\|^2$
and we assume that $\epsilon_t = (\epsilon_t(i))_{i=1}^p,\ t=1,\cdots,n$ is such that $\{\epsilon_t(i)\}$ are independent in $i=1,\cdots,p$ but long-range dependent
in $t$ zero-mean Gaussian time series. Namely, suppose that 
\begin{equation}\label{e:C-epsilon}
C_\epsilon(t):={\rm Cov}(\epsilon_{t+s}(i),\epsilon_s(i)) \sim \sigma^2 |t|^{2H-2},\ \ \mbox{ as }t\to \infty, 
\end{equation}
where $H\in (1/2,1)$ is the Hurst long-range dependence parameter of the time-series.  Consequently, the auto-covariance $C_\epsilon(t)$
is non-summable, i.e.,
$$
\sum_{t=0}^\infty |C_\epsilon(t)| = \infty.
$$

We will need the following auxiliary lemma. 
\begin{lem}\label{le:op norm C_n,e} Let $C_{n,\epsilon}:= (C_\epsilon(t-s))_{n\times n}$, where $C_\epsilon$ is as in \eqref{e:C-epsilon}.  Then, as $n\to\infty$,
we have
$$
\|C_{n,\epsilon}\| = {\cal O} \Big( n^{(2H-1)\vee 1/2} \Big( 1+ 1_{\{H=3/4\}} \log(n) \Big)\Big)
$$
\end{lem}
\begin{proof} Using the Cauchy-Schwartz inequality, we obtain that for all $a\in\R^n$,
$$
 \Big|\Big[C_{n,\epsilon} a \Big](i) \Big| \le \Big(\sum_{j=1}^n C_\epsilon(i-j)^2\Big)^{1/2} \|a\|.
$$
With a standard argument involving the Karamata Theorem and \eqref{e:C-epsilon}, we obtain that
\begin{align*}
    \sum_{j=1}^n C_\epsilon(i-j)^2 & = {\cal O} \Big(\sum_{t=1}^n t^{4(H-1)} \Big) \\ &= {\cal O} \Big( n^{4H-3} (1+ \log(n) 1_{\{H=3/4\}}) \Big). 
\end{align*}
This completes the proof.
\end{proof}

With this lemma in place, we are now ready to obtain an upper bound on $\E\|\hat\Sigma_\epsilon\|^2$.

\begin{lem}\label{le:hat Sigma-epsilon} Let $\epsilon_t = (\epsilon_t(i))_{i=1}^p,\ t=1,\cdots,n$ be
such that $\epsilon_t(i),\ i=1,\cdots,p$ are independent in $i$ zero-mean Gaussian random variables, with their covariance structure satisfying \eqref{e:C-epsilon}.  Denote by
$C_{n,\epsilon}:= (C_{\epsilon}(t-s))_{n\times n}$ the covariance matrix of the vector $(\epsilon_t(1))_{t=1}^n$. Let also
$\widehat \Sigma_\epsilon = n^{-1}\sum_{t=1}^n \epsilon_t \epsilon_t^\top$.  Then,
\[
\mathbb{E}\|\hat \Sigma_{\epsilon}\|^2=O\left(p^2\sigma^4\right),
\]
as $n\to\infty.$
\end{lem}
\begin{proof} Let $E$ be a $p\times n$ matrix whose columns are the $\epsilon_t$'s, i.e., 
$E:= (\epsilon_1\, \cdots\, \epsilon_n)$.  
Then $\widehat \Sigma_\epsilon = n^{-1} E E^\top$. Notice that the eigenvalues of the matrices $E E^\top$ and $E^\top E$ 
are the same and in fact, their largest eigenvalue equals their operator norm with respect to the Euclidean norm. Thus,
$$
\|\widehat \Sigma_\epsilon\|=\sup_{a\not=0} \frac{\|\widehat \Sigma_\epsilon a \|_2}{\|a\|_2} = \frac{1}{n}\lambda_{\max} \Big( E E^\top \Big) = \frac{1}{n}\lambda_{\max} \Big( E^\top E \Big),
$$
where $\lambda_{\max}:=\lambda_{\max}(EE^\top) = \lambda_{\max}(E^\top E)$ is the largest eigenvalue.  Now, however, 
letting $g_i\, i=1,\cdots,p$ denote the columns ot $E^\top$, i.e., $g_i(t) := \epsilon_t(i)$, 
we get
$$
E^\top E = \sum_{i=1}^p g_i g_i^\top,
$$
where the $g_i$'s are independent in $i$.

Consider the Karhunen-Lo\'eve representation of the $g_i$'s, i.e.,
$$
g_i = \sum_{t=1}^n \sqrt{\nu_{n,t}} \varphi_{n,t} Z_{i,t},
$$
where $\varphi_{n,t},\ t=1,\cdots,n$ are orthonormal vectors in $\R^n$, the $Z_{i,t}$'s are independent standard Normal random variables and $\nu_{n,t},\ t=1,2,\hdots,n$ are the (descending) eigenvalues of $C_{n,\epsilon}.$ 
Then, letting $a = \sum_{t=1}^p a_t \varphi_{n,t}$, we obtain 
\begin{align}\label{e:lambda_max_bound}
\lambda_{\max}^2 &= \sup_{\|a\|_2=1} \|E^{\top}Ea\|_2^2 = \sup_{\|a\|_2=1} a^{\top}E^{\top}EE^{\top}Ea\nonumber \\ & =\sup_{\|a\|_2=1} (a^{\top}g_1,\hdots,a^{\top}g_p) EE^{\top}\begin{pmatrix}
    g_1^{\top}a\\ \vdots\\g_p^{\top}a
\end{pmatrix}\nonumber\\ 
& =\sup_{\|a\|_2=1} (a^{\top}g_1,\hdots,a^{\top}g_p) (g_1^{\top}g_1+\hdots+g_p^{\top}g_p)\begin{pmatrix}
    g_1^{\top}a\\ \vdots\\g_p^{\top}a
\end{pmatrix}\nonumber\\ & = \sup_{\|a\|_2=1}\sum_{i=1}^p \langle g_i,g_i\rangle\sum_{j=1}^p\langle g_j,a\rangle^2\nonumber.
\end{align}
Substituting $a$ and $g_i$'s, we have that 
\[\lambda_{\max}^2(EE^{\top}) = \sup_{\|a\|_2=1}\sum_{i=1}^p\left(\sum_{t=1}^n\nu_{n,t}Z_{i,t}^2\right)\sum_{j=1}^p\left(\sum_{k=1}^n\sqrt{\nu_{n,k}}a_kZ_{j,k}\right)^2.\]
Moving the supremum on the inside, one can get that 

\begin{align*}
    \|\hat\Sigma_{\epsilon}\|^2 & = \frac{1}{n^2}\sum_{i=1}^p\left(\sum_{t=1}^n\nu_{n,t}Z_{i,t}^2\right)\sum_{j=1}^p\sup_{\|a\|_2=1}\left(\sum_{k=1}^n\sqrt{\nu_{n,k}}a_kZ_{j,k}\right)^2\\
    & =: \frac{1}{n^2}A\cdot B .
\end{align*}

Thus, using the Cauchy-Schwarz inequality, we obtain that 
\begin{align*}
    \mathbb{E} \|\hat\Sigma_{\epsilon}\|^2\le \frac{1}{n^2}\sqrt{\mathbb{E}A^2}\cdot \sqrt{\mathbb{E}B^2}
\end{align*}
For term $A$, one has
\begin{align*}
    A & = \sum_{i=1}^p\left(\sum_{t=1}^n\nu_{n,t} Z_{i,t}^2\right)\\ & \stackrel{d}{=}\sum_{i=1}^p\tilde \chi_i^2\left((\nu_{n,1},\hdots,\nu_{n,n}),(1,1,\hdots,1),(0,0,\hdots 0)\right),
\end{align*}
where $\tilde \chi_i^2$ are independent generalized chi-squared random variables. So, we have that 
\begin{align*}
    \mathbb{E}A^2 & = {\rm Var}(A)+(\mathbb{E}A)^2 = p^2\cdot {\rm tr}(C_{n,\epsilon})^2+2p\sum_{t=1}^n\nu_{n,t}^2\\ & \le p^2\cdot {\rm tr}(C_{n,\epsilon})^2+2pn\|C_{n,\epsilon}\|_{\rm op}^2.
\end{align*}

For term $B$ we first look at the supremum. Using Lagrange multipliers, we have that 
\begin{align*}
\sup_{\|a\|_2=1}\left(\sum_{k=1}^n\sqrt{\nu_{n,k}}a_kZ_{j,k}\right)^2 = \sum_{k=1}^n\nu_{n,k}Z_{j,k}^2
\end{align*}
which means that 
\begin{align*}
    B & = \sum_{j=1}^p\sum_{k=1}^n\nu_{n,k}Z_{j,k}^2 \stackrel{d}{=} \sum_{k=1}^n\nu_{n,k}\chi_k^2(p),
\end{align*}
where we have used the independence of all $Z_{j,k}$ and $\left(\chi_k^2(p)\right)_{k=1}^n$ are independent chi-square random variables with $p$ degrees of freedom. Thus, 
\begin{align*}
    \mathbb{E}B^2 & = p^2{\rm tr}(C_{n,\epsilon})^2+2p\sum_{k=1}^n\nu_{n,k}^2  \\ & \le p^2{\rm tr}(C_{n,\epsilon})^2+2pn\|C_{n,\epsilon}\|^2.
\end{align*}
Using the bound from Lemma \ref{le:op norm C_n,e} and also ${\rm tr}(C_{n,\epsilon})\le n\cdot \sigma^2,$ we have that 
\begin{align*}
    \mathbb{E}\|\hat \Sigma_{\epsilon}\|^2 & =O\Big(p^2\sigma^4 +2pn^{(2H-2)\vee(-\frac{1}{2})}(1+1_{H=3/4}\log(n))\Big)\\ & = O(p^2\sigma^4),
\end{align*}
as $n\to\infty.$
\end{proof}

Finally, we state a bound on the limit $\mathbb{E}\|\widehat \Sigma-\Sigma\|^2$, assuming that $k$ is fixed or $k=o(p).$

\begin{prop}\label{prop:bound on Sigma hat minus Sigma}
Let $u_t$ be long-range dependent with i.i.d. standard normal marginals independent of $\epsilon_t$, with its covariance structure $\Sigma_u$ satisfying an analogue to Condition \eqref{e:C-epsilon}. Assume also that $\sup_{t}\mathbb{E}\|u_t\|_2^2\le c_u,$ and that $\|\vec{f}_t\|^2\le c_f^2k,$ $ t=1,\hdots, n$. Then,
 \begin{align*}     \lim_{n\to\infty}\mathbb{E}\|\widehat \Sigma-\Sigma\|^2 & = O\Big(p^4 +p^3+ p^2+p\sqrt{{\rm tr}(\Sigma_\epsilon){\rm tr}(\Sigma_u)}\Big).
 \end{align*}
\end{prop}
\begin{proof}
Recall that 
\begin{align*}
    \widehat \Sigma & = \Sigma + B R_f B^\top + \hat\Sigma_u + \hat \Sigma_\epsilon + B R + R^\top B^\top\\ &+ \frac{1}{n}\sum_{t=1}^n \epsilon_tu_t^{\top}+\frac{1}{n}\sum_{t=1}^n u_t \epsilon_t^{\top},    
\end{align*}
where 
\begin{align*}
    R_f &=\frac{1}{n}\sum_{t=1}^n\left(f_tf_t^{\top}-{\cal I}_k\right)\quad {\rm and}\\
    R &=\frac{1}{n}\sum_{t=1}^nf_t(u_t+\epsilon_t)^{\top}.\\
\end{align*}
By Lemma \ref{le:hat Sigma-epsilon}, we immediately have that 
\[\lim_{n\to\infty}\mathbb{E}\|\widehat \Sigma_{\epsilon}\|_{\rm op}^2=O\left(p^2\sigma^4\right).\]

For the term $BR_fB^{\top},$ we know that $\|B\|\le \sqrt{pk} C$. Moreover, we have that $\|f_t\|_2^2\le c_f^2 k$. Thus, \begin{align*}
    \mathbb{E}\|BR_fB^{\top}\|^2 &\le p^2k^2 C^4 \mathbb{E}\|R_f\|^2 \\ &\le p^2k^2C^4\mathbb{E}\left(\frac{1}{n}\sum_{t=1}^n\|{\cal I}_k-f_tf_t^{\top}\|\right)^2\\ & \le  p^2k^2C^4(1+c_f^2k)^2.
\end{align*}
Next, we are looking at $\widehat \Sigma_u.$ Again by Lemma \ref{le:hat Sigma-epsilon}, we have that  $\lim_{n\to\infty}\E\|\widehat\Sigma_u\|^2 = O(p^2\sigma_u^4)$. We now look at the term $BR$. Start with $R$. Then, 
\begin{align*}
    \|R\|& \le \frac{1}{n} \sum_{t=1}^n \|f_t\|_2\|u_t+\epsilon_t\|_2
    \\ &\le \frac{1}{n}\sqrt{k}c_f\sum_{t=1}^n\left(\|u_t\|_2+\|\epsilon_t\|_{2}\right)
\end{align*}
Thus, \begin{align*}
    \mathbb{E}\|BR\|^2\le \frac{1}{n^2}p k^2 C^2c_f^2\E\left(\sum_{t=1}^n(\|u_t\|_2+\|\epsilon_t\|)\right)^2.
\end{align*}
Looking at the expectation above, and because of the independence of $\epsilon_t$ and $u_t$, we have 
\begin{align*}
\E\left(\sum_{t=1}^n(\|u_t\|_2+\|\epsilon_t\|)\right)^2 & = \E\left(\sum_{t=1}^n\|\epsilon_t\|_2\right)^2+\E\left(\sum_{t=1}^n\|u_t\|_2\right)^2\\ & \qquad + 2\sum_{t=1}^n\E\|\epsilon_t\|_2\sum_{t=1}^n\E\|u_t\|_2.
\end{align*}
We have the following bounds
\begin{align*}
    \mathbb{E}\|\epsilon_t\|_2 &\le \sqrt{{\rm tr}(C_{n,\epsilon}})\\
    \mathbb{E}\|u_t\|_2 &\le \sqrt{{\rm tr}(C_{n,u}}). 
\end{align*}
Moreover, by Cauchy-Schwarz, 
\begin{align*}
    \frac{1}{n^2}\mathbb{E}\left[\sum_{t=1}^n \|\epsilon_t\|_2\right]^2 &\le \frac{1}{n^2}\mathbb{E}\left[\sum_{t=1}^n \|\epsilon_t\|_2^2\right]^2 = p^2+\frac{2p}{n},
\end{align*}
since $\sum_{t=1}^n \|\epsilon_t\|_2^2\sim \chi^2(n\cdot p).$ Similarly,
\begin{align*}
    \frac{1}{n^2}\mathbb{E}\left[\sum_{t=1}^n \|u_t\|_2\right]^2 &\le  p^2+\frac{2p}{n},
\end{align*}
Thus, 
\begin{align*}    & \lim_{n\to\infty}\mathbb{E}\|BR\|^2 = O\left(pk^2C^2c_f^2\left[2p^2+2\sqrt{{\rm tr}(\Sigma_u){\rm tr}(\Sigma_{\epsilon})}\right]\right).
\end{align*}
Finally, for the term $\frac{1}{n}\sum_{t=1}^nu_t\epsilon_t^{\top}$, we use the fact that \[\|\epsilon_tu_t^{\top}\|= \max\{u_t^{\top}\epsilon_t,0\}.\] 
Then, 
\begin{align*}
    \mathbb{E}\|\frac{1}{n}\sum_{t=1}^nu_t\epsilon_t^{\top}\|^2&\le \frac{1}{n^2}\mathbb{E}\left[\sum_{t=1}^n\|u_t\epsilon_t^{\top}\|_{\rm op}\right]^2\\ & \le \frac{1}{n^2}\mathbb{E}\left[\sum_{t=1}^n|\epsilon_t^{\top}u_t|\right]^2\\ & \le
    \frac{1}{n^2}\mathbb{E}\left[\sum_{t=1}^n\|\epsilon_t\|_2\|u_t\|_2\right]^2\\ & \le
    \frac{1}{n^2}\mathbb{E}\sum_{t=1}^n\|\epsilon_t\|_2^2\sum_{t=1}^n\|u_t\|_2^2 \\ & =    \frac{1}{n^2}\sum_{t=1}^n\mathbb{E}\|\epsilon_t\|_2^2\sum_{t=1}^n\mathbb{E}\|u_t\|_2^2\\ & =\left(p^2+2p\right)^2.
\end{align*}

Putting everything together, one has that 
\begin{align*}     \lim_{n\to\infty}\mathbb{E}\|\widehat \Sigma-\Sigma\|^2 & = O\Big(p^4 +p^3+ p^2+p\sqrt{{\rm tr}(\Sigma_\epsilon){\rm tr}(\Sigma_u)}\Big).
 \end{align*}
\end{proof}

% and

% \begin{align*}
%     \E\|\vec{r} - \vec{\epsilon} - \vec{u}\|^2 \le \left[\sqrt{\E\|I\|^2}+\sqrt{\E\|I\!I\|^2}+\sqrt{\E\|I\!I\!I\|^2}\right]^2\\\le \Bigg[ \frac{2c_fCk\sqrt{kp}(\E\|\widehat\Sigma-\Sigma\|^2)^{1/2}}{\phi(p)}  +\frac{2\sqrt{kp}(\E\|\widehat\Sigma-\Sigma\|^2)^{1/2}}{\phi(p)} +\sqrt{k} \\ + \sqrt{\min\{k\|\Sigma_u\|,\mbox{tr}(\Sigma_u)\}} + \frac{2\sqrt{k}(\E\|\widehat\Sigma-\Sigma\|^2)^{1/2}\sqrt{\mbox{tr}(\Sigma_u)}}{\phi(p)}\Bigg]^2 \\
%     = \Bigg[(\E\|\widehat\Sigma-\Sigma\|^2)^{1/2}\frac{2\sqrt{pk}}{\phi(p)}\left(c_fCk+1+\sqrt{\frac{\mbox{tr}(\Sigma_u)}{p}}\right)\\ + \sqrt{k}+\sqrt{\min\{k\|\Sigma_u\|,\mbox{tr}(\Sigma_u)\}}\Bigg]^2
% \end{align*}

\end{proof}

\subsection{Bounds on difference of projection operators}

Let $\sigma_1\ge\dots\ge\sigma_k$ be the singular values of $\widehat{B_0}^\top B_0$. The $k$ dimensional vector $(\cos^{-1}(\sigma_1),\dots,\cos^{-1}(\sigma_k))$, called the principal angles, are a generalization of the acute angles between two vectors. Let $\Theta(\widehat{B_0},B_0)$ be the $k\times k$ diagonal matrix with $j$th diagonal entry the $j$th principal angle, then a measure of distance between the space spanned by $\widehat{B_0}$ and ${B_0}$ is $\sin\Theta(\widehat{B_0},B_0)$, where $\sin$ is taken entry-wise.

The following variant of Davis-Kahan theorem \cite{yu2015useful} provides a bound for the distances between eigenspaces.

\begin{thm}[Variant of Davis-Kahan sin $\theta$ theorem] \label{thm:dk}
Let $\Sigma, \widehat{\Sigma} \in \mathbb R^{p\times p}$ be symmetric, with eigenvalues $\lambda_1 \ge \ldots \ge \lambda_p$ and $\widehat\lambda_1 \ge \ldots \ge \widehat\lambda_p$ respectively.
Fix $1 \le k \le p$ and assume that $\lambda_k - \lambda_{k+1}) > 0$, where we define $\lambda_{p+1}=-\infty$. 
Let $V = (\vec{\nu}_1,\vec{\nu}_2, \dots, \vec{\nu}_{k})\in\mathbb R^{p\times k}$ 
and $\widehat V = (\widehat{\vec{\nu}}_1,\widehat{\vec{\nu}}_2, \dots, \widehat{\vec{\nu}}_{k})\in\mathbb R^{p\times k}$ have orthonormal columns satisfying $\Sigma\vec{\nu}_j=\lambda_j\vec{\nu}_j$ and $\widehat\Sigma\widehat{\vec{\nu}}_j=\widehat\lambda_j\widehat{\vec{\nu}}_j$ for $j = 1, 2, \dots, k$. Then
$$
\|\sin\Theta(\widehat{V}, V)\|_\mathrm{F} \le \frac{2{k^{1/2}\|\widehat{\Sigma}-\Sigma\|}}{\lambda_k - \lambda_{k+1}}
$$
\end{thm}

For our purpose, a more meaningful measure of distances between subspaces is the operator norm of the difference of the projections. Let $\mathcal P_{\widehat{B_0}}$, $\mathcal P_{B_0}$ be projections onto the column spaces of $\widehat{B_0}$ and $B_0$ respectively, we have the following

\begin{cor} \label{cl:bound} Under assumptions \eqref{A1} and \eqref{A2} with $\widehat B_0$, $B_0$, 
$\widehat \Sigma$ and $\Sigma$ as defined in \eqref{eq:Sigma-hat}, we have
\begin{equation}
  \E  \|\mathcal P_{\widehat{B_0}}-\mathcal P_{B_0}\|^2 \le \frac{4k\E \|\widehat \Sigma - \Sigma\|^2}{\lambda_k^2}.
\end{equation}
\end{cor}

\begin{proof}[Proof of Corollary \ref{cl:bound} ]
It can be shown that the $l_2$ operator norm of the difference between two projection matrices $\mathcal P_{\widehat{B_0}}$ and $\mathcal P_{{B_0}}$ is determined by the maximum principal angle between the two subspaces (see, e.g. \cite{meyer2000matrix} Chapter 5.15). 
That is,
\begin{equation*}
    \|\mathcal P_{\widehat{B_0}} - \mathcal P_{{B_0}}\| 
    = \sin(\arccos(\sigma_k))
    = \sqrt{1-\sigma_k^2}.
\end{equation*}
Consequently, we have
\begin{equation}
    \|\mathcal P_{\widehat{B_0}}-\mathcal P_{{B_0}}\|^2
    \le \Big(\sum_{j=1}^k 1-\sigma_j^2\Big)
    = \|\sin\Theta(\widehat{B_0},{B_0})\|_\mathrm{F}^2 % \fbox{\color{red} Not sure we need the next 2 lines.} \\
    % &\le \Big(k\big(1-\sigma_{\mathrm{min}}^2(\widehat{B_0}'{B_0})\big)\Big)^{1/2} \\
    % &= k^{1/2}\|\mathcal P_{\widehat{B_0}}-\mathcal P_{{B_0}}\|
\end{equation}
Then applying Theorem \ref{thm:dk} and using the fact that $\lambda_{k+1}=0$ completes the proof.
\end{proof}

\subsection{Consistency of EWMA updated marginal
variance}\label{sec: appendix-consistency}
In this section we present the proof of Proposition~\ref{prop:consistency}. 

Let $\{r_t,\ t\in\mathbb{Z}\}$ be a zero-mean stationary time series, with some correlation structure
\[\rho_t = {\rm Cor}(r_t,r_0) = {\rm Cor}(r_{t+h},r_h),\ t\in\mathbb{Z},\ \forall h\in\mathbb{Z}.\]

We want to handle the following problem: we want to estimate the unknown variance of $X_k$ non-parametrically and show that our estimator is consistent. For the estimation of this variance, we implement an additional EWMA on the squares of the zero-mean stationary series $\{X_k,\ k\in\mathbb{Z}\}.$ Our proposed estimator is 
\begin{equation}\label{e:a-var_estim}
    \hat\sigma_t^2 = (1-\lambda_\sigma)\hat\sigma_{t-1}^2+\lambda_\sigma r_t^2.
\end{equation}
To prove the consistency of this estimator, it suffices to show that the following two properties hold
\begin{align*}
    \mathbb{E}\left[\hat\sigma_t^2\right] & \stackrel{\substack{t\to\infty\\\lambda_\sigma\to0}}{\longrightarrow}\sigma_r^2 := {\rm Var}(r_t)\\
    {\rm Var}\left(\hat\sigma_t^2\right) & \stackrel{\substack{t\to\infty\\\lambda_\sigma\to0}}{\longrightarrow}0.
\end{align*} 
Starting with the expectation, we have that \begin{align*}
    \mathbb{E}\left[\hat\sigma_t^2\right] & = \sum_{j=0}^t\lambda_\sigma(1-\lambda_\sigma)^j\mathbb{E}\left[r_{t-j}^2\right]= \sigma_r^2\sum_{j=0}^t\lambda_\sigma(1-\lambda_\sigma)^j\\ & = \sigma_{r}^2\lambda_\sigma\frac{1-(1-\lambda_\sigma)^{t+1}}{1-(1-\lambda_\sigma)} = \sigma_r^2\left(1-(1-\lambda_\sigma)^{t+1}\right)\\ 
    &\stackrel{t\to\infty}{\longrightarrow}\sigma_r^2,
\end{align*}
since $0<\lambda_\sigma<1$. Here we have used the expression 
\[\hat\sigma_t^2  = \sum_{j=0}^t\lambda_\sigma(1-\lambda_\sigma)^jr_{t-j}^2,\]
since we have a finite horizon on this EWMA.  

Our second goal is to show that the variance of $\hat \sigma_t^2$ vanishes. We start by finding an explicit expression of this variance. We have in general that 
\begin{align*}
    {\rm Var}\left(\hat \sigma_t^2\right) & = {\rm Var}\left(\sum_{j=0}^t\lambda_\sigma(1-\lambda_\sigma)^jr_{t-j}^2\right)\\
    & = \sum_{j=0}^t\lambda_\sigma^2(1-\lambda_\sigma)^{2j}{\rm Var}\left(r_{t-j}^2\right) \\ &+\mathop{\sum_{i=0}^{t}\sum_{j=0}^{t}}_{i\neq j}\lambda_\sigma^2(1-\lambda_\sigma)^{i+j}{\rm Cov}\left(r_{t-j}^2,r_{t-i}^2\right).
\end{align*}

We make the following core assumption to continue our calculations. 
\begin{assum}\label{as:Gaussianity}
The process $\{r_t,\ t\in\mathbb{Z}\}$ is Gaussian.
\end{assum}
Using the fact that $\mathbb{E}(r_t)=0,$ one can immediately obtain that 
\[{\rm Var}\left(r_{t-j}^2\right) = \mathbb{E}\left[r_{t-j}^4\right]-\mathbb{E}\left[r_{t-j}^2\right]^2=3\sigma_r^4-\sigma_r^4=2\sigma_r^4.\]
Now, we need to explore the covariance in the second summand of the above expression. Let 
\begin{align*}
    \begin{pmatrix}
    Z_1\\Z_2
    \end{pmatrix}\sim N\left(0,\begin{pmatrix}
    1 & \rho \\\rho & 1
    \end{pmatrix}\right).
\end{align*}
We know that 
\[Z_1 \stackrel{d}{=} \rho Z_2 +\sqrt{1-\rho^2}Z_3,\]
where $Z_2$ and $Z_3$ are independent standard Normal random variables.. Then, we have that 
\begin{align*}
    & {\rm Cov}  \left(Z_1^2,Z_2^2\right)  \\ &\quad=\mathbb{E}\left[Z_1^2Z_2^2\right]-\mathbb{E}\left[Z_1^2\right]\cdot \mathbb{E}\left[Z_2^2\right] = \mathbb{E}\left[Z_1^2Z_2^2\right]-1 \\ & \quad=  \mathbb{E}\left[\left(\rho^2Z_2^2+(1-\rho^2)Z_3^2+2\rho\sqrt{1-\rho^2}Z_2Z_3\right)Z_2^2\right]]-1\\ & \quad= 3\rho^2+(1-\rho^2)-1 = 2\rho^2.
\end{align*}
Using the above one immediately has that 
\begin{align*}
    {\rm Var}& \left(\hat \sigma_t^2\right) \\  & = 2\sigma_t^4\sum_{j=0}^t\lambda_\sigma^2(1-\lambda_\sigma)^{2j} \\ &\qquad +2\sigma_r^4\mathop{\sum_{i=0}^{t}\sum_{j=0}^{t}}_{i\neq j}\lambda_\sigma^2(1-\lambda_\sigma)^{i+j}\left[{\rm Cor}\left(r_{t-j},r_{t-i}\right)\right]^2\\
    & = 2\sigma_r^4\lambda_\sigma^2\frac{1-\left[(1-\lambda_\sigma)^2\right]^{t+1}}{1-(1-\lambda_\sigma)^2} \\ & \qquad+2\sigma_r^4\lambda_\sigma^2\mathop{\sum_{i=0}^{t}\sum_{j=0}^{t}}_{i\neq j} (1-\lambda_\sigma)^{i+j}\rho_{i-j}^2\\ & =2\sigma_r^4\lambda_\sigma\cdot\frac{1-\left[(1-\lambda_\sigma)^2\right]^{t+1}}{2-\lambda_\sigma}\\ & \qquad+2\sigma_r^4\lambda_\sigma^2\mathop{\sum_{i=0}^{t}\sum_{j=0}^{t}}_{i\neq j} (1-\lambda_\sigma)^{i+j}\rho_{i-j}^2,
\end{align*}
where in the second equality we used the stationarity of $\{r_t,\ t\in\mathbb{Z}\}.$

We take the limit as $t\to \infty$ in the above expression, and we have that 
\begin{align*}
    \lim_{t\to\infty}{\rm Var}\left(\hat \sigma_t^2\right) & = \frac{2\sigma_r^4\lambda_\sigma}{2-\lambda_\sigma} +2\lambda_\sigma^2\sigma_r^4\mathop{\sum_{i=0}^{\infty}\sum_{j=0}^{\infty}}_{i\neq j} (1-\lambda_\sigma)^{i+j}\rho_{i-j}^2.
\end{align*}
We need a condition on $\{\rho_t,\ t\in\mathbb{Z}\}$ in order to proceed.

We propose the following condition in order to secure the consistency of the variance.  
\begin{cond}\label{cond:summability}
Assume that the sequence $\{\rho_t,\ t\in\mathbb{Z}\}$ is square summable, namely that 
\[\sum_{t=-\infty}^{\infty}\rho_t^2<\infty.\]
\end{cond}

Under this condition, we have the following proposition.
\begin{prop}\label{prop:appendix-consistency}
Let $\{r_t,\ t\in\mathbb{Z}\}$ and $\{\hat\sigma_t^2,\ t\in\mathbb{Z}\}$ be defined as in the start of Section~\ref{sec: appendix-consistency} and \eqref{e:a-var_estim} respectively. If Condition~\ref{cond:summability} holds, then the estimator $\hat \sigma_t^2$ is consistent.
\end{prop}
\begin{proof}
The proof is based on the discussion preceding Proposition~\ref{prop:appendix-consistency} and the following. Let \[A_t:= 2\lambda_\sigma^2\sigma_t^4\mathop{\sum_{i=0}^{t}\sum_{j=0}^{t}}_{i\neq j} (1-\lambda_\sigma)^{i+j}\rho_{i-j}^2.\]
Then
\begin{align*}
    A_t & \le 2\lambda_\sigma^2\sigma_t^4\mathop{\sum_{i=0}^{t}\sum_{j=0}^{t}} (1-\lambda_\sigma)^{i+j}\rho_{i-j}^2
    \\ & = 2\lambda_\sigma^2\sigma_r^4\sum_{v=0}^{2k}(1-\lambda_\sigma)^v\sum_{n=\max\{-t,v-2t\}}^{\min\{t,v\}}\rho_n^2\\ & \le 2\lambda_\sigma^2\sigma_r^4\sum_{v=0}^{2t}(1-\lambda_\sigma)^v\sum_{n=-\infty}^{\infty}\rho_n^2\\ &\le 2\lambda_\sigma^2\sigma_r^4\sum_{v=0}^{\infty}(1-\lambda_\sigma)^v\sum_{n=-\infty}^{\infty}\rho_n^2 = 2\lambda_\sigma\sigma_r^4\sum_{n=-\infty}^{\infty}\rho_n^2,
\end{align*}
where we have used the change of variables $v=i+j,n=i-j$ in the second equality. Using the relationship
\begin{align*}
{\rm Var}\left(\hat \sigma_t^2\right)  & = 2\sigma_r^4\lambda_\sigma\cdot\frac{1-\left[(1-\lambda_\sigma)^2\right]^{t+1}}{2-\lambda_\sigma} \\ &+2\lambda_\sigma^2\mathop{\sum_{i=0}^{t}\sum_{j=0}^{t}}_{i\neq j} (1-\lambda_\sigma)^{i+j}\rho_{i-j}^2,    
\end{align*}
we have that 
\[{\rm Var}\left(\hat \sigma_t^2\right)\le \frac{2\sigma_r^4\lambda_\sigma}{2-\lambda_\sigma}+2\lambda_\sigma\sigma_r^4\sum_{n=-\infty}^{\infty}\rho_n^2.\]
Because of Condition~\ref{cond:summability}, letting $\lambda_\sigma\to 0$, we obtain the desired property, namely that ${\rm Var}\left(\hat \sigma_k^2\right)$ vanishes.
\end{proof}

\end{document}